\documentclass[12pt]{article}
\usepackage{amssymb}
\usepackage{amsmath}
\usepackage{amsfonts}
\usepackage{natbib}
\usepackage{setspace}
\usepackage{geometry}
\usepackage{graphicx}
\usepackage{hyperref}
\usepackage{appendix}
\usepackage{caption}
\usepackage{subcaption}
\usepackage{float}

\newtheorem{theorem}{Theorem}

\newtheorem{corollary}{Corollary}

\newtheorem{definition}{Definition}

\newtheorem{lemma}{Lemma}

\newtheorem{proposition}[theorem]{Proposition}

\def \indt{\mathbin{\vbox{\baselineskip=0pt\lineskip=0pt
  \moveright2.5pt\hbox{$\|$}
  \hrule height 0.2pt width 10pt}}}

\input{tcilatex}
\begin{document}

\title{The Projection Solution to the Incidental Parameter Problem\thanks{We thank participants at numerous seminar and conference presentations of related work for their  comments, and especially St\'{e}phane Bonhomme, Allan Collard-Wexler, Bo Honor\'{e}, and Francesca Molinari for helpful discussion.}}
\author{Andrew Chesher \\
UCL and CeMMAP\thanks{%
Address: Department of Economics, University College
London, Gower Street, London WC1E 6BT, United Kingdom. Email:
andrew.chesher@ucl.ac.uk.} \and Adam M. Rosen \\
Duke University and CeMMAP\thanks{%
Address: Adam Rosen, Department of Economics, Duke University, 213 Social
Sciences Box 90097, Durham, NC 27708; Email: adam.rosen@duke.edu.} \and %
Yuanqi Zhang \\
UCL and CeMMAP\thanks{%
Address: Department of Economics, University College London, Gower Street,
London WC1E 6BT, United Kingdom. Email: uctpyqz@ucl.ac.uk.}}
\maketitle

\begin{abstract}
This paper introduces a new approach to econometric analysis of nonlinear panel data
models when the number of observations per observational unit is small. In such models the presence of variables that are constant
within, while varying across, units results in an incidental parameter
problem. The approach taken in this paper removes these incidental parameters \textit{via} projection, which produces a correspondence specifying all
combinations of observed variables and within-unit-varying unobserved
heterogeneity that are achievable by choice of some value of the
unit-specific incidental parameters.  With
unit-specific variables removed, there is no need for assumptions concerning
their joint distribution with other variables. The result is an incomplete model which is
typically partially identifying. Identified sets are characterized \textit{via} moment inequalities using tools of random set theory. Examples of application to
static and dynamic models with discrete or continuous outcomes using distribution-free restrictions on within-unit-varying unobserved heterogeneity are
presented.
\end{abstract}

\section{Introduction}

In econometric analysis using \emph{linear} panel data models, unit-specific heterogeneity terms, so-called
\textquotedblleft fixed effects\textquotedblright , can be removed by
differencing. This enables identification and inference without the need for restrictions
on the covariation of fixed effects with observed explanatory variables and
other unobserved variables, and absent parametric distributional restrictions on unobservable heterogeneity.

By contrast, in almost all \emph{nonlinear} panel data models, differencing
of observed variables or functions thereof does not remove fixed effects. This may not be an issue when observational
units deliver many realizations of outcomes, enabling ``Large-$T$''
analysis. In such cases well-behaved estimators of the values of fixed
effects and the values of parameters common to observational units may be
available.

However, when there are few realizations per unit, as is often the case in
economic data, and the values of fixed effects are estimated, estimators of
common parameters can be very poorly behaved. This is the incidental
parameter problem set out in \cite{Neyman/Scott:48} and reviewed in \cite{Lancaster:00}.

With very few exceptions, approaches to this problem in nonlinear panel models
restrict the joint dependence of fixed effects and other
heterogeneity terms. In many empirical settings such restrictions may not be appropriate. For example, in the context of production function
estimation, firm-specific heterogeneity could be related to managerial
ability affecting both the level and variability of output. It may thus be desirable to allow the covariation of firm-specific and within-unit-varying heterogeneity, upon which measures of total factor productivity depend, to be unrestricted. The moment restrictions of \cite{Arellano/Bond:91} used in the analysis of dynamic \textit{linear} panel models have this feature, without imposing parametric restrictions on the distribution of unobservable heterogeneity. The  approach in this paper enables analysis of \textit{nonlinear} panel models that similarly leave the covariation of fixed effects with other heterogeneity unrestricted, and is applicable without imposing parametric distributional restrictions.



The approach taken here solves the incidental parameter problem by removing them from the model \textit{via} projection. The approach is fundamentally different from others, including the functional differencing approach of \cite{Bonhomme:12}. Functional differencing finds moments that are invariant to the conditional distribution of individual effects given covariates, effectively \textquotedblleft
differencing out\textquotedblright\ the \textit{conditional distribution} of individual effects given covariates. The functional differencing approach is only applicable in models with a parametric specification for the distribution of outcomes conditional on covariates and individual effects. The approach in this paper requires no such restrictions, because it instead projects out the individual effects themselves. With
unit-specific heterogeneity removed, restrictions on its joint distributionß with
other variables are irrelevant.


This paper's projection approach allows great flexibility in the restrictions on unobservable heterogeneity for which identification analysis can be conducted, including nonparametric specifications of the distribution of within-unit-varying heterogeneity. The paper demonstrates with examples that feature moment conditions as in conventional GMM analysis, independence restrictions, conditional quantile restrictions, and pairwise exchangeability restrictions. The projection approach is not tied to any specific type of distributional restriction and can admit many possibilities beyond the specific cases considered here.

To explain this it helps to bring some
notation on board. Consider a panel data model for outcomes $Y\equiv (Y_{1},\dots ,Y_{T})$ with
explanatory variables $X\equiv (X_{1},\dots ,X_{T})$.\footnote{It is straightforward to allow $T$ to vary across
units. In many panel applications $t$ is an index for time, but $t$ could index a group, family, classroom, etc.}
Let $U\equiv (U_{1},\dots ,U_{T})$ denote unobserved heterogeneity varying
within units. Let $V$ denote unit-specific unobserved heterogeneity \emph{not%
} varying within units. Each element, $Y_{t}$, $X_{t}$, $V$ and $U_{t}$ can
be multidimensional.\footnote{Specific to each observational-unit $i$ there is $Y_{it}, X_{it}, V_i, U_{it}$. Observational-unit-specific indices $i$ are omitted to simplify notation.}

Panel data models restrict the functional relationship
satisfied by $Y$, $X$, $V$ and $U$, defining sets of feasible
values that these variables can simultaneously take. The approach proposed here
works with the projection of these sets onto the space of $(Y,X,U)$. This is
the set of values of $Y$, $X$, and $U$ that can be achieved by choice of one
or more values of $V$. Various restrictions on the joint distribution of $Y$%
, $X$, and $U$ can then be considered. With $V$ removed, econometric analysis
can proceed with no restrictions on its joint distribution with other
variables.

In the linear model the projection approach delivers the set of compatible
values of $Y$, $X$ and $U$ as follows.
\begin{equation*}
\{(y,x,u):\forall t\in \{2,\dots ,T\}\quad y_{t}-y_{1}=(x_{t}-x_{1})\beta
+u_{t}-u_{1}\}
\end{equation*}
This set is defined by \emph{equalities}, and the model is complete for
differences in the outcome variables. In nonlinear models the set of
compatible values of $Y$, $X$ and $U$ is typically defined by \emph{inequalities} and there may be no nontrivial functions of outcomes for which the
post-projection model is complete.

The paper shows how such projections can be used in identification analysis
of nonlinear panel models using tools of random set theory, previously employed for identification analysis in \cite{Beresteanu/Molchanov/Molinari:09} and \cite{Chesher/Rosen:17}.\footnote{Knowledge of that theory is not required to apply the results.} A novelty of the analysis here is the application of these tools to models that feature restrictions common in panel contexts but that
are inherently absent in cross sectional settings, such as models with dynamics and models with weak exogeneity restrictions. Employing these different types of restrictions for identification analysis using random set theory is new to this paper. Identified sets for model parameters so-obtained are characterized by moment inequalities, enabling estimation and inference using approaches from the recent literature.\footnote{For example, approaches developed in \cite{Andrews/Shi:17}, \cite{Chernozhukov/Chetverikov/Kato:19}, \cite{Bai/Santos/Shaikh:22}, and \cite{Marcoux/Russell/Wan:24} can be used for asymptotic inference with uncountably many conditional moment inequalities, see also the survey \cite{Shi:25}. Characterizations based on a finite number of moment inequalities are amenable to even more approaches, see for example the recent guide \cite{Canay/Illanes/Velez:26}.}

The paper's main contribution is the projection approach to nonlinear panel models, offering a general approach for removal of incidental parameters which is not tied to any one specific kind of model (e.g. binary response) or distributional restriction. Specific examples considered here show how the analysis delivers several contributions to the nonlinear panel literature, including the following.

\begin{enumerate}

\item Strict and weak exogeneity restrictions allowing feedback can both be accommodated. This speaks to the emphasis in \cite{Chamberlain:22}, \cite{Bonhomme/Dano/Graham:23}, and \cite{Bonhomme:25} on
the importance of relaxing strict exogeneity in panel models.\footnote{\cite{Chamberlain:22} is a posthumously published version of a 1993 working paper.} Recent developments in panel models with weak exogeneity include extension of the functional differencing approach of \cite{Bonhomme:12} to nonlinear models with weak exogeneity in \cite{Bonhomme/Dano/Graham:25}, and partial identification of functionals of the distribution of heterogeneous individual-specific coefficients in linear panel models in \cite{Lee:26}. The analysis in this paper contributes by allowing for weak exogeneity without placing any restrictions on the joint distribution of individual effects and $t$-varying heterogeneity, for example through the use of moment restrictions as in Section \ref{Sec: Identified sets}. 


\item There is flexible treatment of initial conditions in dynamic models, distinct from the random effects treatment in \cite{honore2006bounds}. Distributional restrictions can be conditional or unconditional on the initial condition. Unlike previous treatments, initial conditions can be unobserved since they are then unit-specific heterogeneity terms that can be removed using this paper's projection approach, as demonstrated in Section \ref{Section: Binary Outcome Panels}.

\item Dynamic models in which values of discrete outcomes
depend upon lagged \emph{latent} variables can be accommodated. This is
useful in models in which continuous unobserved variables are coded into
ordered categories as arises, for example, in studies of well-being or
health status, as demonstrated in Section \ref{Section: Ordered Outcome Panels}. Previous papers on ordered response panel models such as \cite{honore2021dynamic} have allowed dependence on \textit{observable} lagged outcomes, but not on lagged latent variables that determine the discrete outcomes.\footnote{As \cite{honore2021dynamic} note, whether it is more appropriate to model lagged dependence on the discrete outcome or on the continuous latent variable depends on the process being studied; see footnote 2 and Appendix D in that paper. We thank Bo Honor\'{e} for calling this to our attention.} 
\end{enumerate}

The paper proceeds as follows. Section \ref{sec: Model and Identification} sets out the class of panel models studied and defines a projection of the set of feasible values of $Y$, $X$, $V$ and $U$ onto the space of $Y$, $X$ and $U$. Identified sets of structures are characterized using random level sets of this projection.

Section \ref{Sec: continuous outcomes} presents two examples of panel data models for continuous outcomes
and derives identified sets of structural features under restrictions on the
correlations amongst $t$-varying heterogeneity and covariates. One example
involves a linear model with censored outcomes, covariates, or both. The
other example concerns a CES production function in which the elasticity of
substitution which appears in the production function in a nonlinear fashion
is firm-specific. The moment restrictions lead to characterizations of
identified sets in terms of Aumann expectations of random sets, and the use of their support functions to obtain moment inequalities.

Section \ref{Sec: Discrete Outcomes} presents examples of dynamic panel data models for discrete
outcomes and shows how to obtain identified sets for common parameters. This leads to characterizations of identified
sets of structural features defined by Artstein's inequalities. There is a
novel treatment of unobserved initial conditions and other missing values. A
new method for accommodating autoregressive latent indexes in dynamic
ordered outcome models is proposed, this by contrast to the commonly
employed approaches in which outcomes follow an autoregressive process. Bounds on parameters of a dynamic binary outcome panel model are derived under quantile independence restrictions and under a conditional exchangeability restriction on the distribution of $U$ absent a parametric specification of that distribution. Section \ref{Sec: Related Lit Discrete Outcomes} discusses the related literature on binary and ordered panel models, and the approaches used to deal with incidental parameters in those models.

Section \ref{Sec: Numerical Illustrations} presents numerical illustrations for the CES production function model introduced in Section \ref{Sec: continuous outcomes} under both strict and weak exogeneity restrictions. Section \ref{sec: conclusion} concludes.

\section{The General Approach}

\label{sec: Model and Identification}

This section first lays out the class of models covered, and then provides a
general set identification characterization that will later be specialized
to specific models to produce moment inequalities usable for estimation and
inference.

\subsection{Model, Notation, and Sampling Process}

Consider a panel model specifying that 
\begin{equation}  \label{panel specification}
Y_{t} \in f(Y^{t-1},X_{t},V,U_{t}),\qquad t \in [T] \equiv\{1,...,T\},
\end{equation}
for some fixed and finite $T$, where $Y^{t-1}$ denotes the vector of values
of $Y_s$ for $s<t$.\footnote{%
Static models are accommodated by specifying $f$ invariant with respect to $%
Y^{t-1}$. In dynamic models initial conditions, such as $Y_0$ in a model
with a one period lag, may either be observable, in which case these are
included in $Y^{t-1}$, or unobservable, in which case they are included in $V
$.} Models allowing multi-valued functions $f$ are accommodated by 
\eqref{panel specification}, thus allowing endogenous explanatory variables
and models admitting multiple equilibria. The function $f$, all of whose arguments may be vectors, is restricted to belong to a set of functions $\mathcal{F}$, which
may be parametrically or nonparametrically specified.

The panel models studied here additionally impose restrictions on the joint
distribution of $U$ and $X$. To incorporate such restrictions, notation $G_{U|X}\left(\cdot|x \right)$ is used to denote a conditional distribution
of $U$ conditional on $X = x$, where for any set $\mathcal{S} \subseteq \mathcal{R}_U$, $G_{U|X}\left(\mathcal{S}|x\right)$ denotes the probability of the
event $U \in \mathcal{S}$ given $X = x$. Notation $G_{U|X}\left(\cdot|\cdot\right)$ denotes a collection of conditional distributions for $U$ given $X=x
$ across all possible values of $x \in \mathcal{R}_{X}$. Throughout the paper notation $\mathcal{R}_A$ denotes the support of any random vector $A$.

When considering the restrictions imposed by a model, notation $\mathsf{G}%
_{U|X}$ is used to denote the family of collections $G_{U|X}\left(\cdot|\cdot\right)$ admitted by the model. For instance, if the components of $U$
are restricted to have zero mean conditional on certain components of $X$
then $\mathsf{G}_{U|X}$ contains all such $G_{U|X}\left(\cdot|\cdot\right)$.

A sampling process delivers realizations of $Y$ and $X$ such that their
joint distribution, $F_{YX}$, is identified. Realizations of $V$ and $U
\equiv \left(U_{1},...,U_{T} \right)$ are not observed. The former is
invariant with respect to $t$. Variables $X \equiv
\left(X_{1},...X_{T}\right)$ have components whose covariation with $t$-varying unobservable heterogeneity $U$ is restricted. The underlying probability space on which all variables are defined is assumed nonatomic throughout.\footnote{This is a mild technical requirement on the underlying probability space that ensures convexity of the Aumann expectation of the sets $\mathcal{Q}(Y,X;\theta)$ defined in Section \ref{Sec: continuous outcomes}. It is assured to hold whenever $U$ is absolutely continuously distributed with respect to Lebesgue measure. See \cite{Beresteanu/Molchanov/Molinari:09} for further discussion.}

Notation $\mathcal{M}$ will be used to denote a set of pairs $m$ of
structural functions and collections of conditional distributions $m=(f,G_{U|X}(\cdot|\cdot))$ that satisfy a model's restrictions. The goal of our
identification analysis is to determine which, if any, $m \in \mathcal{M}$ are
capable of producing a distribution $F_{YX}$ of observable variables.

\subsection{Identification Analysis}

Using the notation laid out above, the restrictions so far described are
formally collected in the following restriction.

\noindent \textbf{Restriction Panel Model (PM)}: Euclidean random vectors $
\left( Y,X,V,U\right) $ are defined on a complete nonatomic probability space $\left( \Omega ,\mathsf{L},\mathbb{P}\right) $ endowed with the Borel sets on $\Omega$ such
that (\ref{panel specification}) holds, where $(f,G_{U|X}(\cdot|\cdot))\in \mathcal{M}$  and $V$ belongs to the set $\mathcal{R}_V$. The distribution of $(Y,X)$ is point identified. $\square $

Equivalent to \eqref{panel specification}, $Y \in \mathcal{Y}(U,V,X;f)$
almost surely, where 
\begin{equation}
\mathcal{Y}(u,v,x;f) \equiv \left\{y=(y_1,...,y_T):y_t \in
f(y^{t-1},x_t,v,u_t), \text{ all } t\in [T] \right\}\text{.}
\end{equation}
The following defines the identified set of structures,
denoted $\mathcal{I}(\mathcal{M},F_{YX})$, delivered by a model $\mathcal{M}$
and distribution $F_{YX}$.

\begin{definition}
\label{Def: ID Set definition}  Under Restriction PM the identified set of
pairs of structural functions $f$ and conditional distributions of
unobservable heterogeneity $G_{U|X}(\cdot|\cdot)$ is  
\begin{multline}  \label{ID set definition equation}
\mathcal{I}(\mathcal{M},F_{YX})\equiv \left\{
\left(f,G_{U|X}(\cdot|\cdot)\right)\in \mathcal{M}:F_{Y|X}(\cdot|x)\preceq 
\mathcal{Y}(\tilde{U},\tilde{V},X;f)\right. \text{ for some } (\tilde{U},%
\tilde{V}) \\
\left. \text{where } \tilde{U} \sim G_{U|X}(\cdot|x) \text{ conditional on }%
X=x\quad \text{a.e.}\quad x\in \mathcal{R}_{X}\right\}\text{.}
\end{multline}
\end{definition}

For any random vector $A$ with distribution $F_{A}$ and random set $\mathcal{A}$, $F_{A}\preceq \mathcal{A}$ denotes that $F_{A}$ is selectionable with
respect to the distribution of $\mathcal{A}$.\footnote{The probability distribution of random variable $A$ is selectionable with
respect to the probability distribution of random set $\mathcal{A}$ when
there exists (i) $\tilde{A}$ having the same distribution as $A$, and (ii) $\widetilde{\mathcal{A}}$ having the same distribution as $\mathcal{A}$, both
defined on the same probability space, such that $\mathbb{P}[\tilde{A}\in 
\widetilde{\mathcal{A}}]=1$. See \cite{Molchanov/Molinari:18} Chapter 2 or Definition 2 of \cite{Chesher/Rosen:Handbook2020}.
} The definition states that the identified set comprises those $\left(f,G_{U|X}(\cdot|\cdot)\right)$ pairs in $\mathcal{M}$ for which there exist random
vectors $(\tilde{U},\tilde{V})$ with $\tilde{U}$ following the distributions
of $G_{U|X}(\cdot|\cdot)$ such that the set of outcomes produced by the
structural function $f$, namely $\mathcal{Y}(\tilde{U},\tilde{V},X;f)$
contains a random vector whose conditional distributions given $X$ match those
of the observed conditional distributions $\{F_{Y|X}(\cdot|x):x\in \mathcal{R}_X\}$ almost surely. Then, and only then, the pair $\left(f,G_{U|X}(\cdot|%
\cdot)\right)$ are capable of producing the distribution $F_{YX}$.

Definition \ref{Def: ID Set definition} follows \cite{Chesher/Rosen:Handbook2020}, here incorporating both types of unobservables 
$U$ and $V$, but it is not directly usable for estimation and inference.
Making it usable requires two steps, as follows. First, it is shown that the
unrestricted individual effects can be removed without loss of identifying
power. Second, results from random set theory are used to yield
characterizations that take the form of moment inequalities that
can be used as a basis for estimation and inference.

\subsubsection{Removing Individual Effects}

Individual effects are removed from the model by use of the projection 
\begin{equation}  \label{Projection Definition}
\mathcal{R}_{YXU}(f) \equiv \left\{(y,x,u) \in \mathcal{R}_{YXU}: \exists
v\in\mathcal{R}_V \text{ s.t. } y_t \in f(y^{t-1},x_t,v,u_t), \text{ all } t
\in [T] \right\}\text{,}
\end{equation}
which is the set of values of observable variables and $t$-varying
unobservable heterogeneity that are mutually compatible when the structural
function is $f$.

Level sets of the projection $\mathcal{R}_{YXU}(f)$ obtained by fixing a
subset of the elements $(Y,X,U)$ are used for identification analysis. For
any $(u,x,f) \in \mathcal{R}_{U} \times \mathcal{R}_{X} \times \mathcal{F}$ 
\begin{equation}  \label{Yset definition}
\mathcal{Y}(u,x;f) \equiv \left\{y \in \mathcal{R}_Y: (y,x,u) \in \mathcal{R}%
_{YXU}(f) \right\} \text{,}
\end{equation}
is the set of possible values for $y$ that can occur when $X=x$ and $U=u$
for some realization of $V$. For any $(y,x,f) \in \mathcal{R}_Y \times 
\mathcal{R}_X \times \mathcal{F}$ 
\begin{equation}  \label{Uiset definition}
\mathcal{U}\left(y,x;f\right) \equiv \left\{u\in \mathcal{R}_U:(y,x,u) \in 
\mathcal{R}_{YXU}(f) \right\}\text{,}
\end{equation}
is the set of possible values for $u$ that can occur when $Y=y$ and $X=x$ for some realization of $V$.
These sets are dual to each other in that for any $u,x,y$, and 
$f$: 
\begin{equation}  \label{dual UYset relationship}
y \in \mathcal{Y}(u,x;f) \iff u \in \mathcal{U}\left(y,x;f\right)\text{.}
\end{equation}
There is the following Proposition.

\begin{proposition}
\label{theorem: profile V} Let Restriction PM hold. Then  
\begin{multline}
\mathcal{I}(\mathcal{M},F_{YX})= \left\{
\left(f,G_{U|X}(\cdot|\cdot)\right)\in \mathcal{M}:F_{Y|X}(\cdot|x)\preceq 
\mathcal{Y}(U,X;f)\right. \text{where } U \sim G_{U|X}(\cdot|x) \\
\left. \text{ conditional on }X=x\quad \text{a.e.}\quad x\in \mathcal{R}%
_{X}\right\}\text{,}  \label{Ystar selectionability}
\end{multline}
and equivalently,  
\begin{multline}
\mathcal{I}(\mathcal{M},F_{YX})= \left\{
\left(f,G_{U|X}(\cdot|\cdot)\right)\in \mathcal{M}:G_{U|X}(\cdot|x)\preceq 
\mathcal{U}(Y,X;f)\right. \text{where } Y \sim F_{Y|X}(\cdot|x) \\
\left. \text{ conditional on }X=x\quad \text{a.e.}\quad x\in \mathcal{R}%
_{X}\right\}\text{.}  \label{Ustar selectionability}
\end{multline}
\end{proposition}

Proposition \ref{theorem: profile V} provides high-level characterizations of
the identified set of $\left(f,G_{U|X}(\cdot|\cdot)\right)$ from which
unobservable individual effects $V$ are absent. The random sets $\mathcal{Y}(U,X;f)$ and $\mathcal{U}(Y,X;f)$ comprise, respectively, the set of random variables $Y$ compatible with structural function $f$ and $(U,X)$, and the
set of random variables $U$ possible given knowledge of
observable variables $(Y,X)$ under structural function $f$.\footnote{The selectionability statement requiring $G_{U|X}(\cdot|x)$ to be selectionable with respect to the distribution of $\mathcal{U}(Y,X;f)$ conditional on $X$ almost surely in \eqref{Ustar selectionability} in Proposition \ref{theorem: profile V} is equivalent to requiring that $(U,X)$ is selectionable with respect to the distribution of the random set $ \mathcal{U}(Y,X;f) \times \{X\}$ by Proposition 1 in Appendix B of \cite{Chesher/Rosen:15}. An analagous statement holds regarding selectionability of $F_{Y|X}(\cdot|x)$ with respect to the distribution of $\mathcal{Y}(U,X;f)$ almost surely in \eqref{Ystar selectionability}, and selectionability of $(Y,X)$ with respect to the distribution of $ \mathcal{Y}(U,X;f) \times \{X\}$. } 

When the function $f$ is restricted to a parametric family indexed by a
parameter, say $\theta$, such that $\mathcal{F} = \left\{f_{\theta}:\theta
\in \Theta \right\}$, then the set-valued mappings defined in (\ref{Yset
definition}) and (\ref{Uiset definition}) will be indexed by $\theta$ rather
than $f$, 
with the associated random sets denoted $\mathcal{Y}\left(U,X;\theta
\right)$ and $\mathcal{U}\left(Y,X;\theta \right)$.

\subsubsection{Characterization via Moment Inequalities}

Proposition \ref{theorem: profile V} provides high-level, generally applicable
characterizations of the identified set for $\left(f,G_{U|X}(\cdot|\cdot)\right)$ in which no distributional restrictions are placed on unobservable
individual effects $V$.\footnote{Characterization of the identified set for $f$ or any functional of $\left(f,G_{U|X}(\cdot|\cdot)\right)$ then follows.} To make these
characterizations practicable requires  necessary and sufficient
conditions for the stated selectionability properties usable for estimation and inference. The recent
literature employing random set theory for identification analysis has made
use of such conditions, which can be employed here as well.\footnote{For further details and alternative conditions for guaranteeing this \textit{selectionability} property see Chapter 2 of \cite{Molchanov/Molinari:18}.}
One approach uses Artstein's inequality under the following mild restriction.
\medskip

\noindent \textbf{Restriction RCS}: For all $\left(f,G_{U|X}(\cdot|\cdot)\right) \in \mathcal{M}$, $G_{U|X}\left( \mathsf{cl}\left(\mathcal{U}(y,x;f)\right) \setminus \mathcal{U}(y,x;f) |x\right) = 0$ $a.e.\text{ } (y,x)$,
where $\mathsf{cl}(\cdot)$ and $\cdot\setminus \cdot$ denote the closure of
a set, and the set difference of two sets, respectively. $\square$  \medskip

Restriction RCS holds automatically if $\mathcal{U}(Y,X;f)$ is closed
almost surely. It also applies when that is not so, but
the difference between $\mathcal{U}(Y,X;f)$ and its closure is measure zero almost surely.\footnote{This is a common occurrence in models in which inequalities determine the
value of a limited dependent variable, for example when a discrete outcome is
determined by whether a continuously distributed unobservable variable exceeds a threshold.} In this case the restriction enables characterization of identified sets by applying results from random set theory to the closure of $\mathcal{U}(Y,X;f)$, which is useful since selectionability criteria from random set theory are often stated for random closed sets.  Lemma 1 in the Appendix provides the formal statement.

\subsubsection{Additional Notation}

For any random vector $A$ or realized value $a$ let $A^{\Delta}_{st} \equiv
A_s - A_t $ and $a^{\Delta}_{st} \equiv a_s - a_t$. For any integer $k>0$, notation $\mathbf{0}_{k}$ denotes a zero vector of length $k$. In the definition of any
set, such as $\mathcal{U}(y,x;f)$ in (\ref{Uiset definition}) above, the
support of a variable is omitted when it is clear from context. Notation
expressing suprema and infima of conditional probabilities or expectations
with respect to the conditioning variables are to be understood as \textit{%
essential} suprema and infima, respectively. For any real number $c$, $c^{-}
\equiv \lvert \min\{c,0\} \rvert $ and $c^{+}\equiv \max\{c,0\}$ denote the
negative and positive part of $c$, respectively. For random vectors $A$
and $B$, $A 
\mathbin{\vbox{\baselineskip=0pt\lineskip=0pt
  \moveright2.5pt\hbox{$\|$}
  \hrule height 0.2pt width 10pt}} B$ signifies that $A$ and $B$ are
stochastically independent. For any vectors $a$ and $b$, $a \cdot b$ denotes their dot product.

\section{Models with continuous outcomes}\label{Sec: continuous outcomes}

This section considers two models with essential nonlinearity and continuous
outcomes. It is shown how the projection approach leads to identification
results and so to estimation and inference absent restrictions involving the
distribution of unit-specific variables and absent a parametric
specification of the distribution of within-unit-varying unobservables.
Identification results are obtained under moment restrictions involving within-unit-varying unobservables and functions of covariates.

In Sections \ref{Sec: CES model} and \ref{Sec: Interval Censored Panels} the
results of projection are derived. In Section \ref{Sec: Identified sets} identification analysis using these projections is provided.

\subsection{\label{Sec: CES model}CES Production Function Panel Models}

This is an example of a nonlinear panel model with continuous
outcomes, inspired by Example 3 of \cite{Bonhomme:12}. Let log output $Y_{t}$
of an individual unit (e.g. a plant or firm) at time $t$ be generated by a
constant elasticity of substitution (CES) production function such that

\begin{equation}
Y_{t}=\beta \log g(L_{t},K_{t},\gamma ,S)+ C +U_{t}\text{,} \qquad t\in[T]\text{,}  \label{CES model for y}
\end{equation}
where $V = (S,C)$ is a pair of unit-specific unobservable variables, $C \in \mathbb{R}$, and 
\begin{equation*}
g(l,k,\gamma ,s)\equiv \left( \gamma l^{s}+(1-\gamma )k^{s}\right)
^{1/s},\quad l,k>0,\quad s\in [-\infty ,1], \quad \gamma\in\left(0,1 \right)\text{,}
\end{equation*}
is the CES function with substitution parameter $s$.\footnote{So $\mathcal{R}_S$ is the subset of the extended real line on which all elements are no greater than one.}

Variables $X_{t}\equiv \left( L_{t},K_{t}\right) $ denote labor and capital
inputs at time $t$, $U_{t}$ denotes $t$-varying unobservable variables, and $\theta =(\beta,\gamma)$ are
common parameters with $\beta > 0$.

Thus \eqref{panel specification} holds for $f = f_\theta$ with
\begin{equation}
f_{\theta}(Y^{t-1},X_t,V,U_t) = \beta \log g(L_{t},K_{t},\gamma ,S)+ C +U_{t}\text{.} 	
\end{equation}

This specification is considered briefly in \cite{Bonhomme:12} as an
example of nonlinear regression, with a unit-specific effect ($S$)
entering nonlinearly.\footnote{Equation \eqref{CES model for y} appears under Example 3 as equation (6) in \cite{Bonhomme:12}, in the notation of that paper using the symbol $\sigma $
where we use $S$. The discussion here uses a simplified version in which
there is no low-skilled labor input.}
In that paper $U\equiv (U_{1},\dots ,U_{T})$ is restricted
Gaussian independent of $(X,V)$, where $V$ is denoted by $\alpha$. It is stated on page 1344,
\textquotedblleft Due to the nonlinearity, it does not seem possible to
difference out $\alpha $ in a straightforward way.\textquotedblright 

The projection approach can be applied here, as is now illustrated. The Gaussian restriction proposed in \cite{Bonhomme:12} is not required.


The set of values of $(Y,X,U)$ delivered by the model for some value of $V$
is obtained as follows.

The function $g(l,k,\gamma ,s)$ is a generalized mean, monotone
increasing in $s\in [-\infty ,1]$, bounded with
\begin{equation*}
g(l,k,\gamma ,-\infty) \equiv \min \{l,k\},\qquad
g(l,k,\gamma ,0) \equiv  l^{^\gamma}k^{1-\gamma}\text{.}
\end{equation*}
With $x_t \equiv (x_{t1},x_{t2}) = (l_t,k_t)$ and $v = (s,c)$ define for each $t$:
\begin{equation}\label{CES mt def}
m_t(y,x,\theta ,v )\equiv y_t - \beta \log g(x_{t1},x_{t2},\gamma,s) - c\text{,}
\end{equation}
which is monotone decreasing in each component of $v$.

The set of feasible combinations of $(y,x,u)$ obtained by projection across $v$ is 
\begin{equation}
\mathcal{R}_{YXU}(\theta) = \bigl\{(y,x,u):\exists v \in \left[-\infty,1\right] \times \mathbb{R} \text{ s.t. } u_t = m_t(y,x,\theta,v) \text{ all } t \in [T] \bigr\}\text{.}  \label{CES projection}
\end{equation}
The $U$-level set of this projection for any realizations $y,x$ is 
\begin{equation}\label{U level sets}
\mathcal{U}\left(y,x,\theta\right) = \bigl\{u:\exists v \in \left[-\infty,1\right] \times \mathbb{R} \text{ s.t. } u_t = m_t(y,x,\theta,v) \text{ all } t \in [T] \bigr\}
\text{.}
\end{equation}

If $V$ were observable then realization of $V=v$ would reveal the realization of $U$ as the singleton value with components $m_t(y,x,\theta,v)$, for each $t \in [T]$. The set $\mathcal{U}\left(y,x,\theta\right)$ is a
manifold on $\mathbb{R}^T$ comprising the set of such vectors compatible
with \emph{some} value of unobservable $V$.


For ease of illustration consider the case in which $T=3$. Manifolds $\mathcal{U}(y,x;\theta)$ are illustrated for an example in which $x=(l_t,k_t)$, $t=1,2,3$ are set according to:
\begin{equation*}
l_1 = 0.5,\quad l_2=1,\quad l_3=2,\quad k_1 = 2.2,\quad k_2 = 1.2,\quad k_3=0.7\text{,}	
\end{equation*}
and parameters are set at $\gamma=0.6$ and $\beta=1$. The left panel of Figure \ref{Figure:CES Four Level Sets} depicts eight manifolds $\mathcal{U}(y,x;\theta)$, one each for values of $y$ with $y^{\Delta}_{21}$ and $y^{\Delta}_{31}$ taking values shown in green in the right panel.\footnote{Because the individual effect $C$ enters additively, each manifold $\mathcal{U}\left(y,x,\theta\right)$ is represented in the space of differences $(u^{\Delta}_{21},u^{\Delta}_{31})$ and is fully determined by the values of $y^{\Delta}_{21}$ and $y^{\Delta}_{31}$.}  Figure \ref{Figure:CES Forty Level Sets} shows the set $A(\theta)$ in blue in the left hand panel comprising the union of the sets $\mathcal{U}(y,x;\theta )$ obtained as $y$ takes values such that $y^{\Delta}_{21}$ and $y^{\Delta}_{31}$ belong to the region $B$ shaded in green in the
right hand panel. Because $(Y^{\Delta}_{21},Y^{\Delta}_{31})\in B$ implies that $(U^{\Delta}_{21},U^{\Delta}_{31})\in A(\theta)$, there is for all $x \in \mathcal{R}_X$ the inequality
\begin{equation*} \mathbb{P}[(Y^{\Delta}_{21},Y^{\Delta}_{31})\in B|X=x] \leq \mathbb{P}[(U^{\Delta}_{21},U^{\Delta}_{31})\in A(\theta)|X=x]\text{,} \end{equation*} which restricts the values of $\theta $ compatible with $F_{YX}$ because of the dependence of $A(\theta)$ on $\theta$. Implications of this sort also arise by use of Artstein's inequality.
If $U$ and $X$ are stochastically independent the inequality above becomes
\begin{equation*} \sup_{x \in \mathcal{R}_X}\mathbb{P}[(Y^{\Delta}_{21},Y^{\Delta}_{31})\in B|X=x] \leq \mathbb{P}[(U^{\Delta}_{21},U^{\Delta}_{31})\in A(\theta)]\text{.} \end{equation*}



Section \ref{Sec: Identified sets} characterizes identified sets for the
common parameters obtained in this model under moment restrictions on the product of elements of $U$ and functions or components of $X$. First, the following section presents a second example of a panel model with continuous outcomes.

\begin{figure}[htbp]
    \centering
       \begin{subfigure}{\textwidth}
        \centering
        \includegraphics[trim = {40 0 10 0},scale=1]{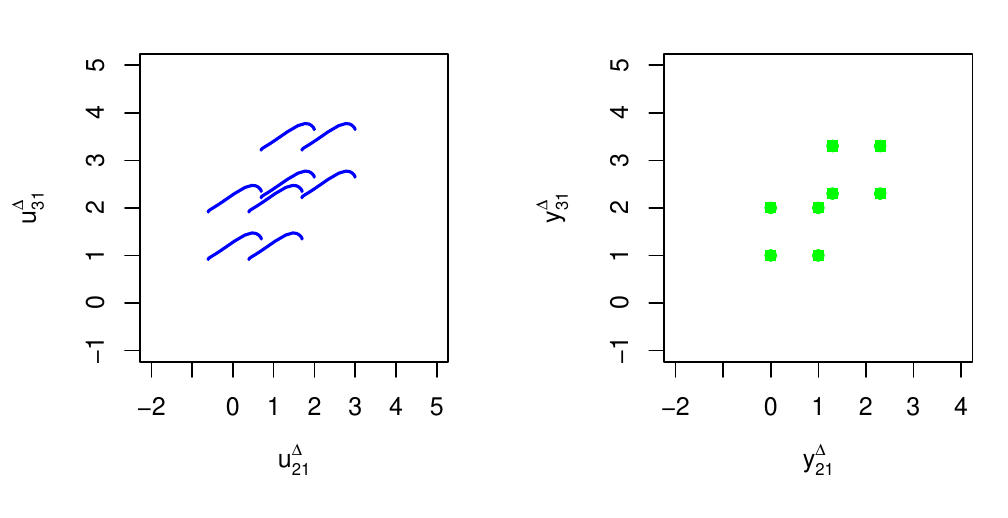} 
        \caption{The panel above on the left illustrates eight manifolds $\mathcal{U}(y,x;\theta)$ for the CES production function example with $x$ and $\theta$ as described in the text. The panel above on the right shows the corresponding values of $y^{\Delta}_{21}$ and $y^{\Delta}_{31}$.} \label{Figure:CES Four Level Sets}
    \end{subfigure}
    \begin{subfigure}{\textwidth}
    \centering
        \includegraphics[trim = {40 0 10 0}, angle = 0, scale = 0.9]{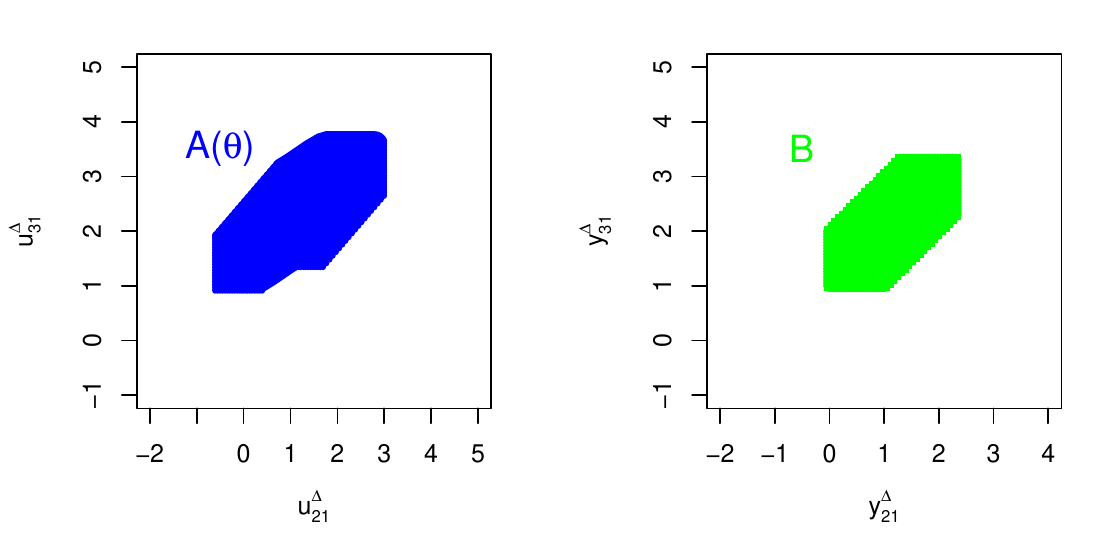} 
        \caption{The panel above on the left illustrates the union of the sets $\mathcal{U}(y,x;\theta)$ as $y$ varies across all values of $y$ such that $y^{\Delta}_{21}$ and $y^{\Delta}_{31}$ take values in the set shown on the right for the CES production function example with $x$ and $\theta$ as described in the text.}\label{Figure:CES Forty Level Sets}
    \end{subfigure}
    \caption{CES production function manifolds $\mathcal{U}(y,x;\theta)$.}
 \end{figure}

\subsection{\label{Sec: Interval Censored Panels}Linear Panels with Censored
Outcomes and Covariates}

This section provides results for linear panel models when data are interval
censored.\footnote{Analysis of panel models with censored outcomes has been studied in e.g. Honor\'{e} (1992, 1993),\nocite{Honore:92} \nocite{Honore:93} \cite{Hu:2002},  \cite{Khan/Ponomareva/Tamer:16}, and \cite{Abrevaya/Muris:20}.}

The model specifies 
\begin{equation}
Y_{t}^{\ast }=X_{t}^{\ast }\theta +C+U_{t},\qquad t\in \lbrack T],
\label{linear outcome model censored}
\end{equation}
where $\theta$ is a $k_x \times 1$ vector of common parameters, each $X_{t}^{\ast }$ is a $1\times k_{x}$ vector, and $V=C$. The unobserved variables are 
$Y^{\ast }$, $X^{\ast }$, $C$, and $U$. The observed variables are 
\begin{equation}
X\equiv \{(X_{t}^{L},X_{t}^{H}):t\in \lbrack T]\},\quad Y\equiv \{\left(
Y_{t}^{L},Y_{t}^{H}\right) :t\in \lbrack T]\}\text{,}
\label{censoring observable vars}
\end{equation}
where 
\begin{equation}
\forall t\in \lbrack T]:\qquad Y_{t}^{L}\leq Y_{t}^{\ast }\leq
Y_{t}^{H},\qquad X_{t}^{L}\leq X_{t}^{\ast }\leq X_{t}^{H}\text{,}
\label{censoring rules}
\end{equation}
and inequalities hold element-wise. There may be interval censoring of
components of one or both of the outcome $Y^{\ast }$ and covariates $X^{\ast
}$. Components that are not censored have $Y_{t}^{L}=Y_{t}^{H}$ for outcomes
and $X_{tj}^{L}=X_{tj}^{H}$ for any uncensored components $X_{tj}^{\ast }$
of $X_{t}^{\ast }$. Missing data can be captured by having both lower and upper
limits correspond to the end points of the support of the corresponding variables.

Define
\begin{align*}
B_{t}^{L}(Y,X;\theta )& \equiv Y_{t}^{L}-\sum_{j=1}^{k_{x}}\max
\{X_{tj}^{H}\theta _{j},X_{tj}^{L}\theta _{j}\}, \\
B_{t}^{H}(Y,X;\theta )& \equiv Y_{t}^{H}-\sum_{j=1}^{k_{x}}\min
\{X_{tj}^{H}\theta _{j},X_{tj}^{L}\theta _{j}\}\text{,}
\end{align*}
There is, for all $s$ and $t$ in $[T]$: 
\begin{align*}
B_{t}^{L}(Y,X;\theta )&\leq C+U_{t} \quad\leq B_{t}^{H}(Y,X;\theta )\text{,} \\
-B_{s}^{H}(Y,X;\theta )&\leq -C-U_{s}\leq -B_{s}^{L}(Y,X;\theta )\text{.}
\end{align*}
Adding the two inequalities yields the projection of the model-admitted set
of values of $(Y,X,C,U)$ onto the space of $(Y,X,U)$ as follows.
\begin{equation}
\mathcal{R}_{YXU}(\theta)=\bigl\{(y,x,u):\forall s,t,\quad
B_{t}^{L}(y,x;\theta )-B_{s}^{H}(y,x;\theta )\leq u_{t}-u_{s}\leq
B_{t}^{H}(y,x;\theta )-B_{s}^{L}(y,x;\theta )\bigr\}
\label{censored model projection}
\end{equation}
Recall that, absent censoring, the linear model by contrast delivers the
projected set 
\begin{equation}\label{censored linear projection}
\mathcal{R}_{YXU}(\theta)=\{(y,x,u):\forall s,t,\quad
u_{t}-u_{s}=y_{t}-y_{s}-(x_{t}-x_{s})\theta \}
\end{equation}
in which there are equalities, whereas with censoring there are inequalities
as in the CES production function case.

With censoring the level set of $U$-values that deliver $Y=y$ when $X=x$ is
simply the slice through the projection defined in (\ref{censored model
projection}) obtained fixing $(y,x)$ accordingly: 
\begin{equation*}
\mathcal{U}(y,x;\theta )=\{u:\forall s,t,\quad B_{t}^{L}(y,x;\theta
)-B_{s}^{H}(y,x;\theta )\leq u_{t}-u_{s}\leq B_{t}^{H}(y,x;\theta
)-B_{s}^{L}(y,x;\theta )\}.
\end{equation*}

This characterization of $\mathcal{U}(Y,X;\theta )$ provides a starting point
for identification analysis using various restrictions on the joint
distribution of $(X,U)$. When covariates are censored,
consideration of context may lead one to prefer restrictions on
the joint distribution of $(X^{\ast },U)$, as considered in the following subsection. 

\subsection{\label{Sec: Identified sets}Identified sets}

Identification analysis for common parameters is now presented for the two
models just considered under moment restrictions on within-unit-varying unobservables and covariates. A characterization for the CES model parameters under a conditional mean restriction is then provided. The methods employed can be applied to more general forms of moment conditions than the ones considered here.

\subsubsection{Moment Restrictions}\label{Sec: Moment Restrictions}

Consider the following restriction. 

\noindent\textbf{Restriction M}. For all $t \in [T]$, $E\left[Z_t U_t \right]=\mathbf{0}_{J_t}$, where each $Z_t = Z_t(X)$ is a vector-valued function of $X$ of dimension $J_t$. 

Weak exogeneity and strict exogeneity of components of $X$ can be accommodated by appropriate definition of each $Z_t$ as will be shown below. Flexibly defining $Z_t$ can further be used to specify that some covariates are strictly exogenous while others are only weakly exogenous.

Since Restriction M requires that $E\left[Z_t U_t \right]=\mathbf{0}_{J_t}$ for all $t$, it is useful to
define the level set of possible values of $\left(\left(Z_1 U_1\right),...,\left(Z_T U_T\right)\right)$ obtained from the projection $\mathcal{R}_{YXU}(\theta)$ defined in \eqref{Projection Definition}, making use of the $U$-level sets defined in \eqref{Uiset definition}:  
\begin{equation}  \label{M level set} \mathcal{Q}(y,x;f) \equiv \left\{ \bigl(\left(z_1  u_1 \right),...,\left(z_T  u_T \right)\bigr):u \in \mathcal{U}(y,x;f)\right\} \text{.}
\end{equation}
Thus $\mathcal{Q}(y,x;f)$ is a
set of vectors, each element of which is a $J \equiv J_1 + \dots +J_T$ dimensional vector whose
entries correspond to feasible values of components of $Z_t U_t$ under structural function $f$ across all $t \in[T]$, when $Y=y$ and $X=x$.

Replacing fixed arguments $(y,x)$ with $(Y,X)$ in \eqref{M level set} yields $\mathcal{Q}(Y,X;f)$, a random set whose distribution is determined by that of $(Y,X)$. Under Restriction M, the panel model specification \eqref{panel specification} with structural function $f$ can produce the distribution of $(Y,X)$ if and only if there is a \textit{measurable selection} of $\mathcal{Q}(Y,X;f)$, defined below, whose expected value is the zero vector $\mathbf{0}_J$.
\begin{definition}
	Let $\mathcal{Q}$ be a random closed set on $\left( \Omega ,\mathsf{L},\mathbb{P}\right)$ whose realizations are subsets of $\mathbb{R}^J$. A random vector $Q$  measurable on $\left( \Omega ,\mathsf{L},\mathbb{P}\right)$ is a \textbf{measurable selection} of $\mathcal{Q}$ if $Q(\omega)\in \mathcal{Q}(\omega) $ for almost all $\omega \in \Omega$.
\end{definition}

Let $\mathbf{L}^{1}(\mathcal{Q})$ denote the set of all integrable measurable selections of a random set $\mathcal{Q}$. The Aumann integral and Aumann expectation of $\mathcal{Q}$ are defined, respectively, as
\begin{equation*}
\mathbb{E}_I[\mathcal{Q}] \equiv \left\{E\left[Q \right]: Q \in \mathbf{L}^{1}(\mathcal{Q})\right\}\text{,}\qquad \mathbb{E}[\mathcal{Q}] \equiv \mathsf{cl}\left(\mathbb{E}_I[\mathcal{Q}]\right)\text{.}	
\end{equation*}
Define the sets
\begin{equation}\label{identified set and moment closure}
\mathcal{F}_I \equiv \left\{f \in \mathcal{F}: \mathbf{0}_J \in \mathbb{E}_I\left[\mathcal{Q}(Y,X;f)\right] \right\},\quad 	\overline{\mathcal{F}_I} \equiv \left\{f \in \mathcal{F}: \mathbf{0}_J \in \mathbb{E}\left[\mathcal{Q}(Y,X;f)\right] \right\}\text{.}
\end{equation}
Under the restrictions of Proposition \ref{Theorem: CES Model Under M}, the set $\mathcal{F}_I$ is the identified set for structural function $f$. The set $\overline{\mathcal{F}_I}$ is a superset of $\mathcal{F}_I$, and therefore provides bounds on $f$.

The set $\overline{\mathcal{F}_I}$ is the \textit{moment-closure} of the identified set in the terminology of \cite{Li:26}, and it is useful for estimation and inference for several reasons. First, the set $\overline{\mathcal{F}_I}$ can coincide with $ \mathcal{F}_I$ and hence be sharp, for example when the Aumann integral is closed, in which case the Aumann integral and Aumann expectation coincide. Such conditions hold under a mild restriction for the CES example of Section \ref{Sec: CES model}, as shown in Proposition \ref{Prop CES sharpness} in the Appendix. Second, characterization of $\overline{\mathcal{F}_I}$ by way of the Aumann expectation is equivalent to a support function characterization which has the form of moment inequalities that only involve expectations of random variables rather than random sets.  Third, there are conditions under which the identified set and its moment closure are statistically indistinguishable even if they do not coincide, as shown in \cite{Li:26}.

The usefulness of the support function characterization for the moment closure follows from two consequences of the restrictions of Proposition \ref{Theorem: CES Model Under M} below. First, there is the equivalence
\begin{equation}\label{support function equivalence}
	\mathbf{0}_J \in \mathbb{E}\left[\mathcal{Q}(Y,X;f)\right] \iff \min_{r\in\mathcal{B}^J}h\left(\mathbb{E}\left[\mathcal{Q}(Y,X;f)\right],r\right) \geq 0\text{,}
\end{equation}
where $\mathcal{B}^J \equiv \left\{r\in \mathbb{R}^J: \lVert r \rVert = 1 \right\}$ is the boundary of the unit ball in $\mathbb{R}^J$ and
\begin{equation*}
	h(\mathcal{Q},r) \equiv \sup_{q \in \mathcal{Q}} r \cdot q\text{,}\qquad  
\end{equation*}
denotes the support function of any set $\mathcal{Q}\subseteq \mathbb{R}^J$ evaluated at $r\in\mathbb{R}^J$.\footnote{This follows from Theorem 2.1.26 of \cite{Molchanov:17} because the underlying probability space is nonatomic under Restriction PM and $\mathcal{Q}(Y,X;\theta)$ is integrable under Restriction M, so $\mathbb{E}\left[\mathcal{Q}(Y,X;\theta)\right]$ is convex.} Second, the order of the expectation and support function can be swapped yielding
\begin{equation*}
	h\left(\mathbb{E}\left[\mathcal{Q}(Y,X;f)\right],r\right) = E\left[h\left(\mathcal{Q}(Y,X;f),r\right) \right]\text{,}
\end{equation*}
where on the right there is the expectation of a random variable.\footnote{This follows from Theorem 2.1.35 of \cite{Molchanov:17} because the underlying probability space is nonatomic.} Putting all this together, there is the following Proposition regarding $\mathcal{F}_I$ and $\overline{\mathcal{F}_I}$ defined in \eqref{identified set and moment closure}.
\begin{proposition}\label{Theorem: CES Model Under M}
	Suppose that Restrictions PM and M hold and that $\mathcal{Q}(Y,X;f)$ is closed almost surely. Then $\mathcal{F}_I$  is the identified set for the structural function $f$. Moreover, the moment closure $\overline{\mathcal{F}_I}$ of $\mathcal{F}_I$ comprises bounds on $f$ and admits the support function representation
	\begin{equation}\label{moment closure set}
		\overline{\mathcal{F}_I} = \left\{f \in \mathcal{F}:\min_{r\in\mathcal{B}^J}E\left[h(\mathcal{Q}(Y,X;f),r) \right] \geq 0\right\}\text{.}
	\end{equation}
	If $\mathbb{E}_I\left[\mathcal{Q}(Y,X;f)\right]$ is closed then $\overline{\mathcal{F}_I} = \mathcal{F}_I$. 
\end{proposition}

Proposition \ref{Theorem: CES Model Under M} states that $\mathcal{F}_I$ is the identified set for $f$, and that its moment closure $\overline{\mathcal{F}_I}$ is characterized by the moment inequalities $E\left[h(\mathcal{Q}(Y,X;f),r) \right] \geq 0$, for all $r\in \mathbb{R}^J$. The reasoning is similar to that of Theorem 4.1 of \cite{Beresteanu/Molchanov/Molinari:09} which established sharp bounds on the best linear predictor with censored outcomes and covariates. The difference is that the analysis here applies to a nonlinear panel model with individual effects, rather than a model for cross section data, so the random set $\mathcal{Q}(Y,X;f)$ is constructed by taking the product of components of $Z$ as specified by Restriction M with the $U$-level set obtained from $ \mathcal{R}_{YXU}(f)$ after removing $V$ from the model by projection.\footnote{The projection step renders an additional integrably boundedness condition used in Theorem 4.1 of \cite{Beresteanu/Molchanov/Molinari:09} inapplicable without further restrictions in the present setting, necessitating the distinction between the identified set and its moment closure.} 

A connection to the support function approach of \cite{Beresteanu/Molchanov/Molinari:09} is also made in \cite{Lee:26} which, in contrast to the projection approach, uses moments that restrict the joint distribution of individual effects with other heterogeneity terms. That paper uses duality theory for infinite dimensional programs to characterize the identified set of features of the distribution of random coefficients in linear panel models. It also shows that if the target parameter is a common parameter, the characterization is equivalent to that obtained by the support function approach used in \cite{Beresteanu/Molchanov/Molinari:09} Theorem 4.1.  Thus, subject to regularity conditions, duality theory for infinite dimensional programs can also be used with the moment conditions of this paper for identification analysis. This relationship also highlights the possibility of using the formulation of \cite{Schennach:14} for developing estimation and inference approaches as a potential alternative to the moment inequality approach, an avenue which is left to future research. 

\subsubsection*{Identification in the CES Model}
In the CES model the class $\mathcal{F}$ is parameterized by $\theta = (\beta,\gamma)$. Focus is given to the moment closure of the identified set for $\theta$, denoted $\Theta^{\ast}$, and the resulting moment inequalities. The notation replaces $f$ with $\theta$ accordingly. In the CES model
\begin{equation}  \label{CES moment set}
\mathcal{Q}(y,x;\theta) \equiv \left\{ \bigl(z_1  m_1\left(y,x;\theta,v) \right),...,\left(z_T  m_T(y,x;\theta,v) \right)\bigr):v\in \left[-\infty,1\right] \times \mathbb{R} \right\} \text{,}
\end{equation}
where $m_t(y,x,\theta,v)$ defined in \eqref{CES mt def} denotes for each $t$ the unique value of $u_t$ given fixed values of $(y,x,\theta,v)$.

The characterization of Proposition \ref{Theorem: CES Model Under M} is specialized under Restriction M with two different specifications for $Z$ as follows.

\noindent\textbf{Restriction MS}. Restriction M holds with $Z_t = (1,X_1,...,X_T)$ for all $t$.\\
\noindent\textbf{Restriction MW}. Restriction M holds with $Z_t = (1,X_1,...,X_t)$ for all $t$.

Consider first Restriction MS which requires that $E\left[X_{hk}U_t \right] = 0$ for all $h,t \in [T]$ and all $k=1,2$ where $X_{h1} = L_h$ and $X_{h2} = K_h$. This is a strict exogeneity restriction comprising $J=T(2T+1)$ moment restrictions. The expectation of the support function of $\mathcal{Q}(Y,X;\theta)$ given by \eqref{CES moment set} under the CES specification simplifies as
\begin{equation*}
	E\left[h(\mathcal{Q}(Y,X;\theta),r) \right] = E\left[\sup_{v} \sum_{t=1}^T w_t(r,X) m_t(Y,X,v,\theta)\right]\text{,}
\end{equation*}
where for each $t\in[T]$,
\begin{equation*}
	w_t(r,x)\equiv r_{t} + \sum_{h=1}^T \sum_{k=1}^2 r_{thk} x_{hk}\text{,}
\end{equation*}
and $r$ is a list of $J = T(2T+1)$ numbers $r_t,r_{thk}$ for all $t,h,k$. 

Using the definition of $m_t$ in \eqref{CES mt def} and simplifying, it follows from Proposition \ref{Theorem: CES Model Under M} that under Restriction MS the set $\Theta^{\ast}$ is the set of $\theta = (\beta,\gamma)$ that satisfy
\begin{equation}\label{CES MS inequalities 2}
\min_{r: \lVert r \rVert = 1 }	 E\left[\sup_{s,c} \sum_{t=1}^Tw_t(r,X)\left(Y_t - \beta \log g(X_{t1},X_{t2},\gamma,s) - c \right)\right] \geq 0\text{.}
\end{equation}

Now suppose instead that only weak exogeneity is asserted, such that Restriction MW is imposed. Working through the same steps under this weaker restriction, Proposition \ref{Theorem: CES Model Under M} delivers $\Theta^{\ast}$ as those parameter vectors satisfying the same inequalities \eqref{CES MS inequalities 2}, but now with $r$ additionally restricted to satisfy $r_{thk}=0$ for all $t<h$. Minimization over this restricted set imposes the zero moment restrictions $E\left[ X_{hk} U_t \right]=0$ only for $t \geq h$.

Section \ref{Sec: Numerical Illustrations} demonstrates that $\Theta^{\ast}$ can produce informative sets by way of numerical illustrations under both weak and strict exogeneity restrictions. In that section further simplification of the inequalities \eqref{CES MS inequalities 2} is provided for that purpose.

Moment restrictions on other functions of covariates and components of $U$ may similarly be imposed through Restriction M beyond the two specific cases of Restriction MS and Restriction MW considered here. Moment conditions incorporating instrumental variables can be used by specifying components of $Z_t$ as functions of components of $X$ with respect to which the structural function is restricted to be invariant, for example by way of exclusion restrictions.

\subsubsection*{Identification in the Censored Linear Panel Model}

Moment restrictions can also be used in the censored linear panel model described in Section \ref{Sec: Interval Censored Panels}. Focus is again given to the moment closure of the identified set for $\theta$, denoted $\Theta^{\ast}$, and the resulting moment inequalities. In this model it may be desirable to invoke moment restrictions involving functions of censored covariate values $X^{\ast}$ rather than the observed endpoints of intervals on which $X^{\ast}$ is realized, which allows censoring to be endogenous.\footnote{If moment restrictions are made solely with respect to $X$, analysis following the steps of the previous section incorporating the linear panel specification applies directly, so is not repeated here.} Thus the following restriction is considered.\\
\noindent\textbf{Restriction M$^{\ast}$}: For all $t \in [T]$, $E\left[Z_t U_t \right]=\mathbf{0}_{J_t}$, where each $Z_t = Z_t(X^{\ast})$ is a vector-valued function of $X^{\ast}$ of dimension $J_t$. 

The level set $\mathcal{Q}(y,x;\theta)$ of possible values of $\bigl(\left(Z_1 U_1\right),...,\left(Z_T U_T\right)\bigr)$ obtained from the censored linear panel projection $\mathcal{R}_{YXU}(\theta)$ given in \eqref{censored linear projection} under Restriction $M^{\ast}$ is
\begin{equation*}
	 \biggl\{\biggl(Z_1(x^{\ast})(y_1^{\ast}-x_1^{\ast}\theta - c),...,Z_T(x^{\ast})(y_T^{\ast}-x_T^{\ast}\theta -c )\biggr): c \in \mathcal{R}_C, (y^{\ast},x^{\ast})\in \mathcal{D}(y,x)\biggr\}\text{,}
\end{equation*}
where $\mathcal{D}(y,x)$ denotes the set of possible values of the censored variables given $(y,x)$:
\begin{equation*}
	\mathcal{D}(y,x) \equiv \bigr\{(y^{\ast},x^{\ast}):  y_{t}^{L}\leq y_{t}^{\ast }\leq
y_{t}^{H},\quad x_{t}^{L}\leq x_{t}^{\ast }\leq x_{t}^{H}, \text{ all } t \in [T]\bigl\}\text{.}
\end{equation*}

Following the same reasoning used when considering Restriction M and the CES model the moment closure of the identified set, $\Theta^{\ast}$, comprises $\theta$ such that $\mathbf{0}_J \in \mathbb{E}\left[\mathcal{Q}(Y,X;\theta)\right]$. Use of the support function yields an equivalent characterization \textit{via} moment inequalities:
\begin{equation*}
	 \min_{r: \lVert r \rVert = 1 }	 E\left[\sup \left\{\sum_{t=1}^T \biggl( \overrightarrow{r_t} \cdot  \bigl( Z_t(x^{\ast})(y_t^{\ast}-x_t^{\ast}\theta - c)\bigr)\biggr): c \in \mathcal{R}_C,\text{ } (y^{\ast},x^{\ast})\in \mathcal{D}(y,x) \right\} \right] \geq 0\text{,}
\end{equation*}
where each $\overrightarrow{r_t}$ is a vector of length $J_t$ such that $r = (\overrightarrow{r_1},...,\overrightarrow{r_T})$.

As was the case for Restriction M, Restriction M$^{\ast}$ can accommodate different moment restrictions through specification of $Z_t$. For example, $Z_t = (1,X^{\ast}_1,...,X^{\ast}_T)$ and $Z_t= (1,X^{\ast}_1,...,X^{\ast}_t)$ for strict and weak exogeneity restrictions, respectively.

\subsubsection{Conditional Moment Restrictions}
Consider the following conditional moment restriction, which provides a strict exogeneity restriction stronger than that of Restriction MS.

\noindent\textbf{Restriction CMS}: For all $t \in [T]$, $E\left[U_t|X \right] = 0$.

Restriction CMS implies that $\mathbf{0}_T \in \mathbb{E}\left[\mathcal{U}(Y,X;\theta) |X\right]$ almost surely, i.e. that the zero vector is an element of the conditional Aumann Expectation of the $U$-level set, the form of which for the CES model is given in \eqref{U level sets}.

Using the support function approach in the CES model this is equivalently that
\begin{equation*}
\min_{\lVert r \rVert =1}E\left[\sup_{s \in [-\infty,1] }\sup_{c \in \mathbb{R}} \sum_{t=1}^Tr_t\bigl(Y_t - \beta \log g(X_{t1},X_{t2},\gamma,s)-c \bigr)|X=x \right] \geq 0, \text{ a.e. } x,
\end{equation*}
where $r=(r_1,...,r_T)$. This can further be expressed
\begin{equation*}
\min_{\lVert r \rVert =1} \left(\sum_{t=1}^T r_t \mu_t(x)  - \min_{s \in [-\infty,1]} \sum_{t=1}^T r_t \beta \log g(x_{t1},x_{t2},\gamma,s) - \inf_{c\in \mathbb{R}} c \sum_{t=1}^T r_t \right) \geq 0, \text{ a.e. } x,
\end{equation*}
where $\mu_t(x) \equiv E\left[Y_t|X \right]$. From this it follows that for any $\theta$ only values of $r_1,...,r_T$ that sum to zero can provide the minimum over $r$, and the inequalities simplify to

\begin{equation*}
\min_{ \{ r: \lVert r \rVert =1, r_1 + \cdots +r_T = 0 \}} \left(\sum_{t=1}^T r_t \mu_t(x)  - \min_{s \in [-\infty,1]} \sum_{t=1}^T r_t \beta \log g(x_{t1},x_{t2},\gamma,s)  \right) \geq 0, \text{ a.e. } x,
\end{equation*}

So, as in the cases considered with unconditional moment restrictions, the bounds are characterized by an infinite collection of moment inequalities, in this case infinitely many conditional moment inequalities. As noted in the introduction estimation and inference methods from the recent literature are available, as discussed for example in the survey \cite{Shi:25}.






\section{Models with discrete outcomes}\label{Sec: Discrete Outcomes}
This section gives examples of application of the projection approach in
static and dynamic panel models with \emph{discrete} outcomes. A dynamic
binary outcome model is studied in Section \ref{Section: Binary Outcome
Panels}; a dynamic ordered response model is studied in Section \ref{Section: Ordered Outcome Panels}.

It is shown how, in models in which ordered discrete outcomes encode the values of continuous latent variables, autoregressive dependence in the latent continuous variables can be accommodated. This approach can also be employed in other discrete response models, such as multinomial choice.

In dynamic models attention must be paid to initial values of outcomes if
they are not observed. Following the approach of this paper, when they are
not observed they are treated as unit-specific unobserved
variables and are removed by projection. This is illustrated in the models considered in both Sections \ref{Section: Binary Outcome
Panels} and \ref{Section: Ordered Outcome Panels}.

Section \ref{Section: Discrete Outcomes Identified Sets} provides characterizations of identified sets for the binary and ordered response models of Sections \ref{Section: Binary Outcome Panels} and \ref{Section: Ordered Outcome Panels}. Moment-based restrictions such as those considered in continuous outcome models in Section \ref{Sec: continuous outcomes} can be uninformative in discrete outcome models; see  \cite{Manski:88}. So here attention is turned to the identifying power of stochastic independence restrictions. Section \ref{Subsec: NP Distributions} considers further options for conducting identification analysis absent a parametric specification of the distribution of $t$-varying unobservables, such as quantile and exchangeability restrictions.

Section \ref{Sec: Related Lit Discrete Outcomes} provides discussion of the related literature on binary and ordered outcome panel models. In contrast to the approach here, nearly all such models in the literature restrict the joint distribution of unit-specific effects and within-unit-varying heterogeneity.\footnote{The analysis in \cite{aristodemou2021semiparametric} is the sole exception of which we are aware.} To our knowledge, no previous models allow for unobserved initial conditions or for dependence on lagged \textit{latent} variables.

\subsection{\label{Section: Binary Outcome Panels}A Dynamic Binary Outcome
Panel Model}

Consider the following binary outcome panel specification. 
\begin{equation}
Y_{t}=1\left[ X_{t}\beta +\gamma Y_{t-1}+C+U_{t}\geq 0\right] ,\qquad t\in
\lbrack T \rbrack\text{,}  \label{binary model}
\end{equation}
with parameter vector $\theta =(\gamma ,\beta ^{\prime })$, $\mathcal{R}_C = \mathbb{R}$, and $U$ continuously distributed with full support on $\mathbb{R}^T$ conditional on $X$.\footnote{This implies that Restriction RCS holds and there is no loss in using weak inequalities throughout in expressions for the sets $\mathcal{R}_{YXU}(\theta)$ and $\mathcal{U}(y,x;\theta)$.} Higher order lags are easily accommodated. 

In a dynamic model in which $\gamma $ may be nonzero, there is an initial
condition to be considered. If the value of $Y_{0}$ is observable, then the $t$-invariant individual unobservable is $V=C$, while if $Y_{0}$ is not
observable $V=(C,Y_{0})$. Whichever is the case, define $\mathcal{Y}_{0}$ to
be the set of possible values of the initial condition such that, if the
initial condition is observable, then $\mathcal{Y}_{0}=\{y_{0}\}$ and if it
is not then $\mathcal{Y}_{0}=\{0,1\}$.

Projecting away individual effects yields
\begin{equation*}
\mathcal{R}_{YXU}(\theta) = \left\{(y,x,u): \exists (y_0,c) \in \mathcal{Y}_0 \times \mathbb{R}:y_t = 1\left[x_t \beta + \gamma y_{t-1} + c +u_t \geq 0 \right] , \text{ } t\in[T]  \right\}	\text{.}
\end{equation*}
It is convenient to define sets of indices
\begin{equation*}
\mathcal{T}_{0}\equiv \left\{ t\in \lbrack T]:Y_{t}=0\right\} \text{,}\qquad 
\mathcal{T}_{1}\equiv \left\{ t\in \lbrack T]:Y_{t}=1\right\} \text{.}
\end{equation*}
Observing that
\begin{equation*}
\begin{array}{cccc}
\forall t\in \mathcal{T}_{0}\text{:}\quad \quad & X_{t}\beta +\gamma
Y_{t-1}+U_{t}\leq & -C &  \\ 
\forall t\in \mathcal{T}_{1}\text{:}\quad \quad &  & -C & \leq X_{t}\beta
+\gamma Y_{t-1}+U_{t}\text{,}
\end{array}
\end{equation*}
the projection can be expressed as
\begin{equation*}
\mathcal{R}_{YXU}(\theta) = \left\{(y,x,u): \exists y_0 \in \mathcal{Y}_0 \text{ s.t. }\max_{t\in \mathcal{T}_{0}}\{x_{t}\beta +\gamma
y_{t-1}+u_{t}\}\leq \min_{t\in \mathcal{T}_{1}}\{x_{t}\beta +\gamma
y_{t-1}+u_{t}\}  \right\}	\text{.}
\end{equation*}

The $U$-level set of this projection for any $(y,x)$ is 
\begin{equation}
\mathcal{U}\left( y,x;\theta \right) =\left\{ u:\exists y_{0}\in \mathcal{Y}%
_{0}\text{ s.t. }\max_{t\in \mathcal{T}_{0}}\{x_{t}\beta +\gamma
y_{t-1}+u_{t}\}\leq \min_{t\in \mathcal{T}_{1}}\{x_{t}\beta +\gamma
y_{t-1}+u_{t}\}\right\}\text{.}  \label{egref}
\end{equation}
In static models with $\gamma =0$ there is the further simplification
\begin{equation*}
\mathcal{U}\left( y,x;\theta \right) =\left\{ u:\max_{t\in \mathcal{T}%
_{0}}\{x_{t}\beta +u_{t}\}\leq \min_{t\in \mathcal{T}_{1}}\{x_{t}\beta
+u_{t}\}\right\} \text{.}
\end{equation*}

If observations in any periods $t$ are missing for a class of units, for
example if there is an unbalanced panel, then such $t$ are in neither $
\mathcal{T}_{0}$ nor $\mathcal{T}_{1}$ and the set $\mathcal{U}(y,x;\theta )$
leaves the value of $U_{t}$ unrestricted. The identification
analysis here still applies, and will result in inequalities that reflect
the lack of restrictions on $U_{t}$ in such periods. In a dynamic model with
the value of some $Y_{t}$ not observed that value becomes an additional
unobserved unit-specific variable in the determination of $Y_{t+1}$ removed \textit{via} projection along with other such variables. Unbalanced panels can be handled in the same manner.

To illustrate identification analysis in dynamic binary response panel
models consider the following example.\medskip

\noindent \textbf{Example 1: Two and three period binary response}.\medskip 

When $T=2$, if $Y_{0} = y_0$ is observed then $
\mathcal{Y}_{0}=\{y_{0}\}$ and (\ref{egref}) simplifies as follows. 
\begin{gather*}
\mathcal{U}((y_{0,}0,0),x;\theta )=\mathcal{R}_{U}\text{,} \\
\mathcal{U}((y_{0,}0,1),x;\theta )=\left\{ u:u_{21}^{\Delta }\geq
-x_{21}^{\Delta }\beta +\gamma y_{0}\right\}\text{,} \\
\mathcal{U}((y_{0,}1,0),x;\theta )=\left\{ u:u_{21}^{\Delta }\leq
-x_{21}^{\Delta }\beta +\gamma (y_{0}-1)\right\}\text{,} \\
\mathcal{U}((y_{0,}1,1),x;\theta )=\mathcal{R}_{U}\text{.}
\end{gather*}
Identification regions for $\theta $ using the inequalities arising when $Y=(0,1)$ and when $Y=(1,0)$ are provided by \cite{aristodemou2021semiparametric} for models in which $Y_{0}$ is observed and
either one of $U\mathbin{\vbox{\baselineskip=0pt\lineskip=0pt
  \moveright2.5pt\hbox{$\|$}
  \hrule height 0.2pt width 10pt}}X\mid Y_{0}$ or $U\mathbin{\vbox{\baselineskip=0pt\lineskip=0pt
  \moveright2.5pt\hbox{$\|$}
  \hrule height 0.2pt width 10pt}}X$ hold. \cite{khan2023identification}
provide sharp identification regions for $\theta $ in dynamic binary
response models for arbitrary finite $T$ under a conditional stationarity
restriction, with $Y_{0}$ observed.

By contrast, taking the projection approach, it is not necessary to have $Y_{0}$ observed. When $Y_{0}$ is not observed the set $\mathcal{U}
(y,x,\theta )$ is simply the union of the sets $\mathcal{U}(y,x;\theta )$ obtained on setting $y_{0}=0$ and then $y_{0}=1$. The sets
corresponding to $Y\in \{(0,0),(1,1)\}$ are unchanged; the others are as
follows. 
\begin{equation*}
	\mathcal{U}((0,1),x,\theta ) =\left\{ u:u_{21}^{\Delta }\geq
-x_{21}^{\Delta }\beta -\gamma ^{-}\right\}, \quad \mathcal{U}((1,0),x,\theta ) =\left\{ u:u_{21}^{\Delta }\leq
-x_{21}^{\Delta }\beta +\gamma ^{-}\right\}\text{,}
\end{equation*}
where $\gamma^- \equiv -\min\{\gamma,0\}$. Table \ref{Table: Binary T 3 Y0 observed} shows the sets $\mathcal{U}(y,x;\theta )$ for the case in which $T=3$ and $Y_{0}$ is observable.\footnote{The inequalities that appear here and in similar models involving threshold
crossing conditions and linear indexes are routine to derive using
Fourier-Motzkin elimination.} Table \ref{Table: Binary T 3 Y0 not observed} shows the sets when $T=3$ and $Y_{0}$ is
not observed.

For the case in which $T=3$ the sets $\mathcal{U}(y,x;\theta )$ can be expressed involving just $u_{31}^{\Delta }$
and $u_{32}^{\Delta }$, since $u_{21}^{\Delta }=u_{31}^{\Delta }-u_{32}^{\Delta }$, enabling visualization on the
space of $(U_{31}^{\Delta },U_{32}^{\Delta })$. The six nontrivial sets $\mathcal{U}(y,x;\theta )$ for each $y\in \mathcal{R}_{Y}$ are
depicted in Figure \ref{Figure:DynBinaryUsets} for the case in which $\theta
=(\gamma ,\beta )$ with $\gamma =1$ and $X=x$ such that $-x_{32}^{\Delta }\beta =-2$, and $-x_{31}^{\Delta }\beta =2$. $\square$
\begin{figure}[tbph]
\centering 
\includegraphics[scale = 0.42, trim = {100 20 50 0}, clip]{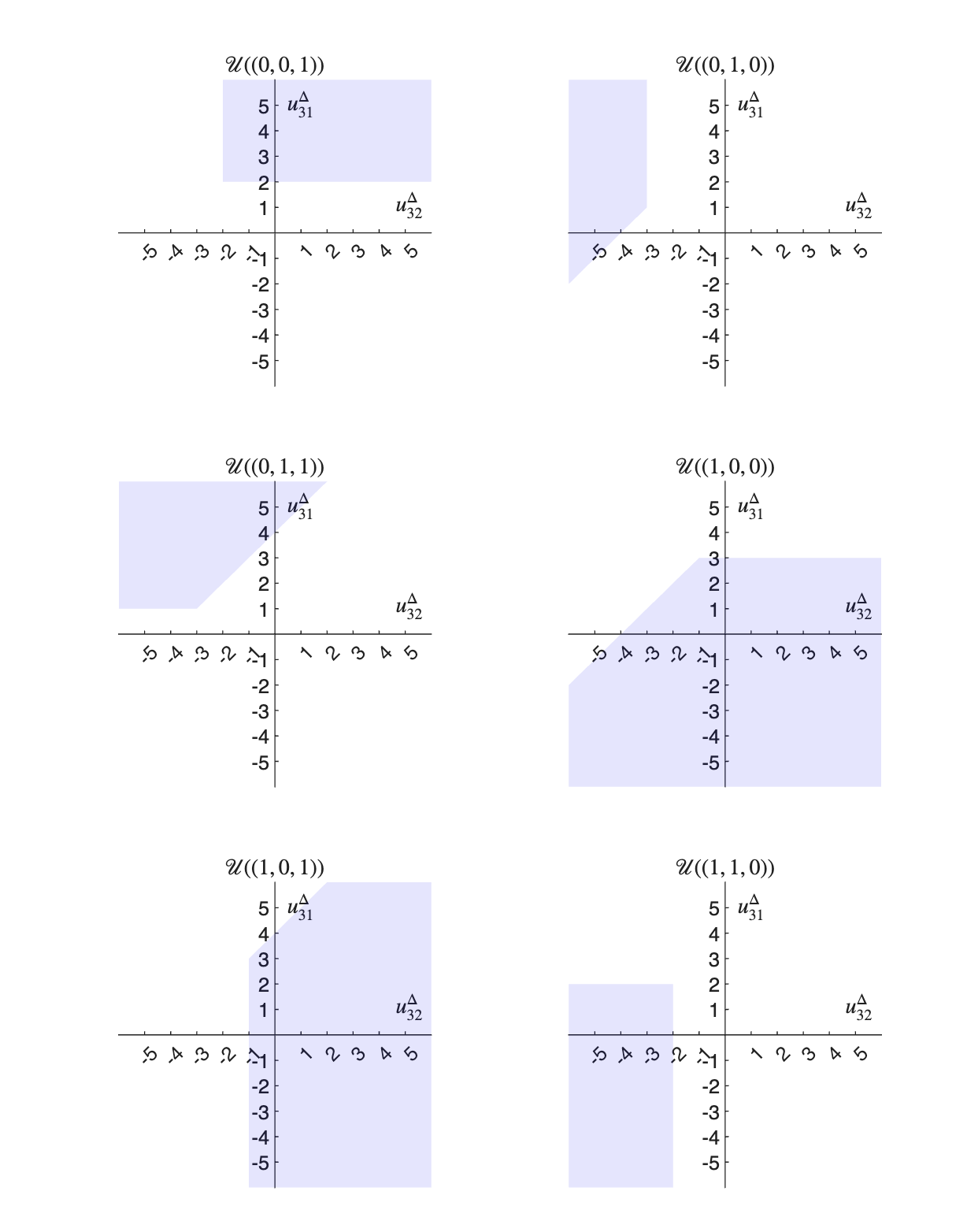}
\caption{The support of $\mathcal{U}(Y,X;\protect\theta)$ in a dynamic
binary response model with $X=x$ and $\protect\theta$ such that $-x^{\Delta}_{32}\protect\beta = -2$, $-x^{\Delta}_{31}\protect\beta = 2$, $\gamma = 1$ and unobservable initial condition, projected onto the space of $\left(U^{\Delta}_{31},U^{
\Delta}_{32}\right)$. Additionally, the sets $\mathcal{U}\left((0,0,0),x;\theta \right)$ and $\mathcal{U}\left((1,1,1 ),x;\theta\right)$ comprise the entire space and are therefore not illustrated.}
\label{Figure:DynBinaryUsets}
\end{figure}

\begin{table}[tbp] \centering%
\caption{$U$-level sets in a dynamic binary response model with $T=3$ and $Y_0$
observed.}\medskip 
\begin{tabular}{|c|c|}
\hline
$y$ & $\mathcal{U}(y,x;\theta )$ \\ \hline\hline
$(0,0,0)$ & $\mathcal{R}_{U}$ \\ \hline
$(0,0,1)$ & $\left\{ u:u_{31}^{\Delta }\geq -x_{31}^{\Delta }\beta +\gamma
y_{0}\text{ }\wedge \text{ }u_{32}^{\Delta }\geq -x_{32}^{\Delta }\beta
\right\} $ \\ \hline
$(0,1,0)$ & $\left\{ u:u_{21}^{\Delta }\geq -x_{21}^{\Delta }\beta +\gamma
y_{0}\text{ }\wedge \text{ }u_{32}^{\Delta }\leq -x_{32}^{\Delta }\beta
-\gamma \right\} $ \\ \hline
$(0,1,1)$ & $\left\{ u:u_{21}^{\Delta }\geq -x_{21}^{\Delta }\beta +\gamma
y_{0}\text{ }\wedge \text{ }u_{31}^{\Delta }\geq -x_{31}^{\Delta }\beta
+\gamma (y_{0}-1)\right\} $ \\ \hline
$(1,0,0)$ & $\left\{ u:u_{21}^{\Delta }\leq -x_{21}^{\Delta }\beta +\gamma
(y_{0}-1)\text{ }\wedge \text{ }u_{31}^{\Delta }\leq -x_{31}^{\Delta }\beta
+\gamma y_{0}\right\} $ \\ \hline
$(1,0,1)$ & $\left\{ u:u_{21}^{\Delta }\leq -x_{21}^{\Delta }\beta +\gamma
(y_{0}-1)\text{ }\wedge \text{ }u_{32}^{\Delta }\geq -x_{32}^{\Delta }\beta
+\gamma \right\} $ \\ \hline
$(1,1,0)$ & $\left\{ u:u_{31}^{\Delta }\leq -x_{31}^{\Delta }\beta +\gamma
(y_{0}-1)\text{ }\wedge \text{ }u_{32}^{\Delta }\leq -x_{32}^{\Delta }\beta
\right\} $ \\ \hline
$(1,1,1)$ & $\mathcal{R}_{U}$ \\ \hline
\end{tabular}%
\label{Table: Binary T 3 Y0 observed}%
\end{table}%

\begin{table}[tbp] \centering%
\caption{$U$-level sets in a dynamic binary response model with $T=3$ and $Y_0$
not observed.}\medskip 
\begin{tabular}{|c|c|}
\hline
$y$ & $\mathcal{U}(y,x;\theta )$ \\ \hline\hline
$(0,0,0)$ & $\mathcal{R}_{U}$ \\ \hline
$(0,0,1)$ & $\left\{ u:u_{31}^{\Delta }\geq -x_{31}^{\Delta }\beta -\gamma
^{-}\wedge u_{32}^{\Delta }\geq -x_{32}^{\Delta }\beta \right\} $ \\ \hline
$(0,1,0)$ & $\left\{ u:u_{21}^{\Delta }\geq -x_{21}^{\Delta }\beta -\gamma
^{-}\wedge u_{32}^{\Delta }\leq -x_{32}^{\Delta }\beta -\gamma \right\} $ \\ 
\hline
$(0,1,1)$ & $\left\{ u:u_{21}^{\Delta }\geq -x_{21}^{\Delta }\beta -\gamma
^{-}\wedge u_{31}^{\Delta }\geq -x_{31}^{\Delta }\beta -\gamma ^{+}\right\} $
\\ \hline
$(1,0,0)$ & $\left\{ u:u_{21}^{\Delta }\leq -x_{21}^{\Delta }\beta +\gamma
^{-}\wedge u_{31}^{\Delta }\leq -x_{31}^{\Delta }\beta +\gamma ^{+}\right\} $
\\ \hline
$(1,0,1)$ & $\left\{ u:u_{21}^{\Delta }\leq -x_{21}^{\Delta }\beta +\gamma
^{-}\wedge u_{32}^{\Delta }\geq -x_{32}^{\Delta }\beta +\gamma \right\} $ \\ 
\hline
$(1,1,0)$ & $\left\{ u:u_{31}^{\Delta }\leq -x_{31}^{\Delta }\beta +\gamma
^{-}\wedge u_{32}^{\Delta }\leq -x_{32}^{\Delta }\beta \right\} $ \\ \hline
$(1,1,1)$ & $\mathcal{R}_{U}$ \\ \hline
\end{tabular}%
\label{Table: Binary T 3 Y0 not observed}%
\end{table}%


\subsection{\label{Section: Ordered Outcome Panels}Ordered response models}

Consider a dynamic ordered response panel model with 
\begin{equation}
Y_{t}=\sum_{j=1}^{J}j\times 1\left[ \alpha _{j}\leq Y_{t}^{\ast }<\alpha
_{j+1}\right] ,\quad Y_{t}^{\ast }=X_{t}\beta +L_{t}\gamma +C+U_{t},\quad
t\in \lbrack T]\text{,}  \label{ordered response structural function}
\end{equation}
where $\alpha _{J+1}=-\alpha _{0}=\infty $, $\theta =\left( \alpha
_{1},...,\alpha _{J},\beta,\gamma\right) $, and $U$ is continuously distributed with full support on $\mathbb{R}^T$ conditional on $X$.\footnote{Thus as in Section \ref{Section: Binary Outcome Panels} there is no loss in using weak inequalities in expressions for $\mathcal{U}(y,x;\theta)$.}  In the static case $\gamma =0$.  

The outcome $Y_{t}$ is ordered categorical, taking the value $j\in
\{0,...,J\}$ if $Y_{t}^{\ast }\in \left[ \alpha _{j},\alpha _{j+1}\right) $,
where $Y_{t}^{\ast }$ is a latent index. The normalization $\alpha _{1}=0$
is imposed since one of the $\alpha _{j}$ parameters can be absorbed by
unit-specific variable $C$. The variable $L_{t}$ here denotes
functions of lagged outcomes\ in a model in which $t$ indexes time. For
example there could be $L_{t}=\iota _{t-1}$ where $\iota _{t-1}\equiv \left( 1\left[ Y_{t-1}=1\right] ,...,1\left[ Y_{t-1}=J
\right] \right)$
in a model with one-period lagged outcome dependence.\footnote{It is straightforward to accommodate multiple lags, for example two lags with $L_{t}=(\iota _{t-1},\iota _{t-2})$. The indicator for one value of $Y_{t-1}$, here $1\{Y_{t-1}=0\}$, is omitted by normalization as in \cite{honore2021dynamic}.} In this example $\gamma
=(\gamma _{1},...,\gamma _{J})'$ so $L_{t}\gamma =\gamma
_{j}$ if and only if $Y_{t-1}=j$.

Now expressions for the sets of values of $U$ that can occur given values of
observed variables are derived. These are the $U$-level sets of the projection $\mathcal{R}_{YXU}(\theta)$ for the ordered response structural function \eqref{ordered response structural function}.

To deal with cases in which the lag variable is not observed, define $\mathcal{L}(y,x;\theta )$ as the set of possible values of the lag variables 
$(L_{1},...,L_{T})$ given the observability of lagged outcomes. In this
exposition only one period lags are considered.

If $L_{t}=\iota _{t-1}$, and the initial value $Y_{0}$ is not
observed, then $\mathcal{L}(y,x;\theta )$ is the set of $\iota
_{0},...,\iota _{T-1}$ produced by observed $y_{1},...,y_{t-1}$ with any
value of $\iota _{0}$. If instead lag dependence manifests through
unobservable realizations $y^{\ast }$ as studied below, then $\mathcal{L}
(y,x;\theta )$ will restrict each $L_{t}$ to the interval implied by
observed realizations $y$.

To obtain the $U$-level set $\mathcal{U}(y,x;\theta )$ note that there is
for all $s$ and $t$%
\begin{eqnarray*}
\alpha _{y_{t}}-x_{t}\beta -l_{t}\gamma -u_{t} &\leq &C\leq \alpha
_{y_{t}+1}-x_{t}\beta -l_{t}\gamma -u_{t} \\
-\alpha _{y_{s}+1}+x_{s}\beta +l_{s}\gamma +u_{s} &\leq &-C\leq -\alpha
_{y_{s}}+x_{s}\beta +l_{s}\gamma +u_{s}
\end{eqnarray*}%
and upon adding%
\begin{equation*}
\alpha _{y_{t}}-\alpha _{y_{s}+1}-x_{ts}^{\Delta }\beta -l_{ts}^{\Delta
}\gamma -u_{ts}^{\Delta }\leq 0\leq \alpha _{y_{t}+1}-\alpha
_{y_{s}}-x_{ts}^{\Delta }\beta -l_{ts}^{\Delta }\gamma -u_{ts}^{\Delta }
\end{equation*}%
which leads to 
\begin{multline*}
\mathcal{U}(y,x;\theta )=\{u:\exists l\in \mathcal{L}\left( y,x;\theta
\right) \text{ s.t. }\forall s<t\in \lbrack T] \\
\alpha _{y_{t}}-\alpha _{y_{s}+1}-x_{ts}^{\Delta }\beta -l_{ts}^{\Delta
}\gamma \leq u_{ts}^{\Delta }\leq \alpha _{y_{t}+1}-\alpha
_{y_{s}}-x_{ts}^{\Delta }\beta -l_{ts}^{\Delta }\gamma \}\text{.}
\end{multline*}

In a static model with $\gamma =0$ there is no lag dependence and the
simplification 
\begin{equation*}
\mathcal{U}(y,x;\theta )=\{u:\forall s<t\in \lbrack T]\quad \alpha
_{y_{t}}-\alpha _{y_{s}+1}-x_{ts}^{\Delta }\beta \leq u_{ts}^{\Delta }\leq
\alpha _{y_{t}+1}-\alpha _{y_{s}}-x_{ts}^{\Delta }\beta \}\text{.}
\end{equation*}

In contrast to other approaches to dynamic ordered response panel models,
unobservable initial conditions can be accommodated using the projection
approach. Moreover, period $t$ lags can include functions of
both observable and latent variables, such as lagged values of the
unobserved index $Y^{\ast }$. This is important because in many applications
the ordered outcome $Y_{t}$ may depend not just on the value of 
$Y_{t-1}$ but on the location of $Y_{t-1}^{\ast }$ relative to the
thresholds $\alpha _{j}$. For example, if $Y_{t}$ is a categorical measure
of health status, the effect of current health on future health may be
different for two individuals in \textquotedblleft good\textquotedblright\
health, one of whom is close to the boundary for the \textquotedblleft
fair\textquotedblright\ health category, and the other close to the boundary
for the \textquotedblleft excellent\textquotedblright\ health category. Or,
in application to letter grades obtained in a sequence of courses, dynamic
impacts may be more effectively measured by a student's numerical score
rather than, say, whether they achieved an \textquotedblleft
A-\textquotedblright\ or \textquotedblleft B+\textquotedblright .

\subsubsection*{Lagged Outcome Dependence}

With one period lagged outcome dependence $
\mathcal{L}(y,x;\theta )=\mathcal{L}_{1}(y,x;\theta )\times \dots \times 
\mathcal{L}_{T}(y,x;\theta )$ where $\mathcal{L}
_{t}(y,x;\theta )=\{\iota _{t-1}\}$ for all $t>1$, and $\mathcal{L}_{1}(y,x;\theta )$ is the set of standard basis vectors in $\mathbb{R}^{J}$.

To illustrate consider such a model with two periods, three categories, $T=J=2$, and $Y_0$ not observed. When $Y=(0,0)$ or $Y=(2,2)$ any value of $U$ is possible, as there
is always $C$ small enough or large enough, respectively, to produce either
outcome. Table \ref{table:ordered response lag dependence U sets TJ2} shows the other $U$-level sets in this model.

\begin{table}[t]
\centering
\begin{tabular}{c|c} 
 \hline 
 \hline 
  $Y$ &  $\mathcal{U}\left(y,x;\theta \right)$\\
 \hline 
 $(0,1)$ &   $\left\{ u:u_{21}^{\Delta }\geq -x_{21}^{\Delta }\beta +\min \{0,\gamma _{1},\gamma _{2}\}\right\}$  \\
$(0,2)$ &  $\left\{ u:u_{21}^{\Delta }\geq \alpha_{2}-x_{21}^{\Delta }\beta +\min \{0,\gamma _{1},\gamma _{2}\}\right\}$  \\
$(1,0)$ & $\left\{ u:u_{21}^{\Delta }\leq
-x_{21}^{\Delta }\beta -\gamma _{1}+\max \{0,\gamma _{1},\gamma
_{2}\}\right\}$ \\
$(1,1)$ &  $\left\{ u:\exists y_{0}\text{ s.t.}-\alpha
_{2}-x_{21}^{\Delta }\beta \leq u_{21}^{\Delta }-\gamma _{y_{0}}+\gamma
_{1}\leq \alpha _{2}-x_{21}^{\Delta }\beta \right\}$ \\
$(1,2)$ &  $\left\{ u:u_{21}^{\Delta }\geq
-x_{21}^{\Delta }\beta -\gamma _{1}+\min \{0,\gamma _{1},\gamma
_{2}\}\right\}$ \\
$(2,0)$ &  $\left\{ u:u_{21}^{\Delta }\leq -\alpha
_{2}-x_{21}^{\Delta }\beta -\gamma _{2}+\max \{0,\gamma _{1},\gamma
_{2}\}\right\}$ \\
$(2,1)$ &  $\left\{ u:u_{21}^{\Delta }\leq
-x_{21}^{\Delta }\beta -\gamma _{2}+\max \{0,\gamma _{1},\gamma
_{2}\}\right\}$ \\
 \hline 
 \hline
\end{tabular}
\caption{$U$-level sets for the dynamic ordered response panel model with lagged outcome dependence, $J=T=2$, and $Y_0$ not observed.}
\label{table:ordered response lag dependence U sets TJ2}
\end{table} 

\subsubsection*{Lagged Latent Dependence}

Now consider the case in which the period $t$ outcome depends on the lagged
latent index $Y_{t-1}^{\ast }$ so $L_{t}=Y_{t-1}^{\ast }$. Repeated
substitution for values of $Y_{s}^{\ast }$, $s<t$, in the equation for $Y_{t}^{\ast }$ in \eqref{ordered response structural function} gives 
\begin{equation*}
Y_{t}^{\ast }=\gamma ^{t}Y_{0}^{\ast }+\sum_{s=1}^{t} \gamma ^{t-s} \bigl(X_{s}\beta
+C+U_{s}\bigr)\text{.}
\end{equation*}
So the set of values of $U$ that deliver $Y=y$ when $X=x$ is 
\begin{equation*}
\mathcal{U}(y,x;\theta )=\left\{ u:\exists (y_{0}^{\ast },c)\text{ s.t. }%
\forall t\in \lbrack T],\text{ }\alpha _{y_{t}}\leq \gamma ^{t}y_{0}^{\ast
}+\sum_{s=1}^{t} \gamma ^{t-s} \bigl(x_{s}\beta
+c+u_{s}\bigr)\leq \alpha _{{y_{t}+1}}\right\} 
\text{.}
\end{equation*}

The set $\mathcal{U}(y,x;\theta )$ differs from the case with lagged outcome
dependence. The inequalities defining $\mathcal{U}(y,x;\theta )$ are linear
in $c$ and $y_{0}^{\ast }$ so Fourier-Motzkin elimination can be used to
remove these variables from the inequalities that define $\mathcal{U}(y,x;\theta )$. 

\subsection{\label{Section: Discrete Outcomes Identified Sets}Identified
sets}

With the $U$-level sets defined as in discrete outcome models such as those presented in Sections \ref{Section: Binary Outcome Panels} and \ref{Section: Ordered Outcome Panels}, moment inequality characterizations of identified sets for common parameters can be obtained using Artstein's inequality.\footnote{Artstein's inequality is established in \cite{artstein1983distributions}, see also \cite{Molchanov:17} pages 83--84 and \cite{Molchanov/Molinari:18} Section 2.2.} The inequality provides the following corollary to Proposition \ref{theorem: profile V}.

\begin{corollary}\label{Containment Corollary} Suppose that Restrictions PM and RCS hold.  Then the identified set for $\left(f,G_{U|X}(\cdot|\cdot)\right)$ comprises those pairs $\left(f,G_{U|X}(\cdot|\cdot)\right)\in \mathcal{M}$ such that
\begin{equation}\label{Artstein's inequality}
\mathbb{P}\left[\mathcal{U}(Y,X;f) \subseteq \mathcal{S} |X=x \right] \leq G_{U|X}\left(\mathcal{S} |x\right), \text{ a.e. } x\in\mathcal{R}_X\text{.}	
\end{equation}
for all closed $\mathcal{S} \subseteq \mathcal{R}_U$. The identified set for $f$ is the set of $f$ such that (\ref{Artstein's inequality}) holds for $f$ and some $G_{U|X}(\cdot|\cdot)\in \mathsf{G}_{U|X}$ with $\left(f,G_{U|X}(\cdot|\cdot)\right)\in \mathcal{M}$.
\end{corollary}

It will now be demonstrated how this corollary can be specialized to produce moment inequality characterizations of identified sets for common parameters. Prior applications of this inequality in the partial identification literature include \cite{Beresteanu/Molchanov/Molinari:10} and \cite{Chesher/Rosen:17}, see \cite{Molinari:Handbook} for further references. The novelty here is not in the use of Artstein's inequality for identification analysis, but rather its application to level sets of the projection $\mathcal{R}_{YXU}$ obtained by removal of incidental parameters, and under distributional restrictions commonly found in panel models having no counterpart in cross section models.

The characterization of the identified set provided by Corollary \ref{Containment Corollary} using Artstein's inequality comprises for each $x \in \mathcal{R}_X$ as many inequalities as the number of closed sets in $\mathcal{R}_U$. Previous papers such as \cite{Galichon/Henry:09}, \cite{Chesher/Rosen:17}, and \cite{Luo/Ponomarev/Wang:25} have characterized \textit{core determining collections} that comprise a smaller collection of sets $\mathcal{S}$ such that if \eqref{Artstein's inequality} holds for all $\mathcal{S}$ in the collection, then it holds for all closed sets. To the best of our knowledge these results have not been previously employed in panel models but they are applicable here to simplify characterizations of identified sets delivered by Artstein's inequality.

With the characterizations of the $U$-level sets of Sections \ref{Section: Binary Outcome Panels} and \ref{Section: Ordered Outcome Panels} the inequality \eqref{Artstein's inequality} delivers observable implications for the common parameters. This is now illustrated in the context of Example 1 in Section \ref{Section: Binary Outcome Panels}. The same steps can be taken to characterize identified sets for the ordered response panel models of Section \ref{Section: Ordered Outcome Panels}.

\noindent\textbf{Example 1, continued}: Consider the dynamic binary panel model with unobservable initial condition and $T=3$. The $U$-level sets for a particular $x$ and $\theta$ are illustrated in Figure \ref{Figure:DynBinaryUsets}. We can see immediately from the figure that with, for example, $\mathcal{S} = \left\{u: u^{\Delta}_{31} \geq u^{\Delta}_{32} \right\}$ the inequality \eqref{Artstein's inequality} becomes $\mathbb{P}\left[Y \in \{(0,1,0),(0,1,1) \} |X=x\right] \leq G_{U|X}\left(\mathcal{S}|x\right)$. This set is however not amongst the minimal core determining collection.

From Theorem 1 of \cite{Chesher/Rosen:17} it follows that only sets $\mathcal{S}$ that comprise unions of sets on the support of $\mathcal{U}(Y,X;f)$ need consideration. Theorem 3 of that paper establishes that among this collection, one need not consider those sets $\mathcal{S}$ that can be partitioned into two sets $\mathcal{S}_1$ and $\mathcal{S}_2$ such that either $\mathcal{U}(Y,X;f) \subseteq \mathcal{S}_1$ or $\mathcal{U}(Y,X;f) \subseteq \mathcal{S}_2$, but not both simultaneously. Such a set $\mathcal{S}$ is not \textit{self-connected} in the terminology of \cite{Luo/Ponomarev/Wang:25}.\footnote{That paper also shows that it is generally possible to achieve further refinement, establishing that among the class of sets comprising unions of sets on the support of $\mathcal{U}(Y,X;f)$, those that are both self-connected and complement-connected comprise a minimal core-determining collection. However, in the panel models studied here in which $\mathcal{U}(Y,X;f) = \mathcal{R}_U$ for some $y \in \mathcal{R}_Y$, the requirement that sets be complement-connected provides no reduction in the core-determining collection.} The minimal core determining collection of sets for the value of $x$ and $\theta$ that produce the $U$-level sets of Figure \ref{Figure:DynBinaryUsets} yields 32 moment inequalities of the form \eqref{Artstein's inequality}. $\square$
\smallskip

If $U$ and $X$ are stochastically independent there is the following simplification of Artstein's inequality
 \begin{equation}\label{Artstein's inequality independence}
\sup_{x \in \mathcal{R}_X}\mathbb{P}\left[\mathcal{U}(Y,X;f) \subseteq \mathcal{S} |X=x \right] \leq G_{U}\left(\mathcal{S}\right)\text{.}
\end{equation}
This applies with both parametric and nonparametric restrictions on the class of functions $f$ and distributions $G_{U}(\cdot|\cdot)$ admitted by the model.  For example, if the utility function and distribution of unobservable heterogeneity are parametrically specified up to $\theta \in \Theta \subseteq \mathbb{R}^d$ in \eqref{Artstein's inequality independence}, $f$ may be replaced by $\theta$ and $G_U(\mathcal{S})$ by $G_U(\mathcal{S};\theta)$.

Even when using only core determining collections of sets, the number of inequalities can be large. Nonetheless, approaches for asymptotic inference with infinitely many conditional moment inequalities can be used, see for instance Section 2.2 of \cite{Chernozhukov/Chetverikov/Kato:19} and Example 2 of \cite{Andrews/Shi:17}.

The next section shows how identification analysis can proceed under nonparametric specifications of the distribution of within-unit-varying heterogeneity.

\subsection{Nonparametric distributional specifications}\label{Subsec: NP Distributions}

An advantage of the projection approach developed in this paper is that it enables identification analysis when there are neither parametric nor stationarity restrictions on the distribution of
within-unit-varying heterogeneity. In Section \ref{Sec: continuous outcomes} it was shown how this can be
achieved using moment restrictions. Here are two alternative approaches more suited to models
of discrete outcomes. Other such restrictions are possible.

Consider models such as the discrete outcome models considered in this section in which $U$ level sets are
determined entirely by restrictions on differences $u_{st}^{\Delta }$ for a
collection of values of $s$ and $t$.

First consider \emph{quantile independence restrictions}. For chosen values
of $s$ and $t$, first define an ascending sequence of quantile probabilities 
\begin{equation*}
p_{st}\equiv (p_{st}^{0},p_{st}^{1},\dots ,p_{st}^{K},p_{st}^{K+1})
\end{equation*}
which are specified values in $\left[ 0,1\right] $ with $p_{st}^{0}\equiv 0$ and $p_{st}^{K+1}\equiv 1$. Then define additional parameters, namely the
unknown elements of 
\begin{equation*}
\lambda _{st}\equiv \{\lambda _{st}^{0},\lambda _{st}^{1},\dots ,\lambda
_{st}^{K},\lambda _{st}^{K+1}\}
\end{equation*}
with $\lambda _{st}^{0}\equiv -\infty $, $\lambda _{st}^{K+1}\equiv +\infty $. These new parameters are values of the quantiles of the marginal
distributions of the $U_{st}^{\Delta }$ at the chosen quantile
probabilities, e.g. $(0.25,0.5,0.75)$, and there are the restrictions\footnote{It is easy to impose restrictions of  symmetry and unimodality if that were desired.}
\begin{equation*}
\forall x\in \mathcal{R}_{X}\quad \mathbb{P}[U_{st}^{\Delta }\leq \lambda
_{st}^{k}|X=x]\equiv p_{st}^{k}\quad \text{free of }x\text{.}
\end{equation*}

Moment inequalities characterizing the identified set of values of the
parameters $\theta $ and the quantile values in $\lambda _{st}$ are
\begin{equation*}
\forall k\in \{1,\dots ,K+1\}\quad \sup_{x\in \mathcal{R}_{X}}\mathbb{P}[\mathcal{U}(Y,X;\theta )\subseteq \{u:\lambda _{st}^{k-1}\leq u_{st}^{\Delta
}\leq \lambda _{st}^{k}\}|X=x]\leq p_{st}^{k}-p_{st}^{k-1}
\end{equation*}
\begin{equation*}
\forall k\in \{1,\dots ,K\}\quad \sup_{x\in \mathcal{R}_{X}}\mathbb{P}[
\mathcal{U}(Y,X;\theta )\subseteq \{u:u_{st}^{\Delta }\leq \lambda
_{st}^{k}\}|X=x]\leq p_{st}^{k}
\end{equation*}
for all pairs $(s,t)$ for which the quantile independence restrictions are
maintained. Identified sets for $\theta$ are obtained as those values of $\theta$ for which there exist values of $\lambda_{st}$ such that all such inequalities are satisfied.
\cite{Chesher/Kim/Rosen:23} gives details and has an example of this
approach in action in a different, non-panel, context.\footnote{\cite{Chesher/Kim/Rosen:23} studies an IV Tobit model in which explanatory
variables may be endogenous. Proposition 4 in Section 4.2 deals with
quantile independence restrictions and is easily extended to the models
considered in this paper.}

Finally consider \emph{pairwise conditional exchangeability restrictions}
requiring that, for some chosen $s$ and $t$, $U_{s}$ and $U_{t}$ are
exchangeable conditional on $X$. Under this restriction the median of $%
U_{st}^{\Delta }$ conditional on $X=x$ is zero for all $x$. The inequalities
above with $K=1$, $\lambda _{st}^{1}=0$ and $p_{st}^{1}=0.5$ deliver bounds
on $\theta $ absent parametric restrictions on the distribution of $U$.%
\footnote{Under the conditional exchangeability restriction the probability density
function of $U_{st}^{\Delta }$ is symmetric around zero. This may deliver
additional bounds in some cases.} 

\noindent\textbf{Example 1, continued}: Consider again the dynamic binary outcome model as in \eqref{binary model} with unobserved initial value $Y_{0}$ and $T=3$ with $U$-level sets shown in Table \ref{Table: Binary T 3 Y0 not observed}.

The zero median independence restriction implied by pairwise conditional
exchangeability of all elements of $U$ delivers the identified set of values
of $(\beta ,\gamma )$ as those satisfying
\begin{equation}\label{median inequalities}
\bigwedge\limits_{j=1}^{6}\left( \sup_{x\in \mathcal{R}_{X}}w_{j}(x,\beta
,\gamma )\leq \frac{1}{2}\right) 
\end{equation}
where the expressions $w_{j}(x,\beta ,\gamma )$ are shown in Table \ref{Table: CExch}. In this table, $p_{x}(y_{1},y_{2},y_{3})\equiv \mathbb{P}%
[Y=(y_{1},y_{2},y_{3})|X=x]$ and the column headed \textquotedblleft $%
\mathcal{S}$\textquotedblright\ shows the set $\mathcal{S}$ in the zero
median independence restriction 
\begin{equation*}
\forall x\in \mathcal{R}_{X},\quad \mathbb{P}[U\in \mathcal{S}|X=x]=1/2
\end{equation*}
that delivers each row of the table. $\square$

\begin{table}[tbp] \centering%
\caption{Expressions $w_{j}(x,\beta ,\gamma )$ in \eqref{median inequalities} required to be no greater than $1/2$ under a restriction that each $U_s$ and $U_t$ are exchangeable conditional on $X$ with  $T=3$. }\medskip 
\begin{tabular}{|c|c|c|}
\hline
$j$ & $\mathcal{S}$ & $w_{j}(x,\beta ,\gamma )$ \\ \hline\hline
$1$ & $\{u:u_{31}^{\Delta }\leq 0\}$ & \multicolumn{1}{|l|}{$%
p_{x}(1,1,0)\times 1[-x_{31}^{\Delta }\beta +\gamma ^{-}\leq
0]+p_{x}(1,0,0)\times 1[-x_{31}^{\Delta }\beta +\gamma ^{+}\leq 0]$} \\ 
\hline
$2$ & $\{u:u_{31}^{\Delta }\geq 0\}$ & \multicolumn{1}{|l|}{$%
p_{x}(0,0,1)\times 1[-x_{31}^{\Delta }\beta -\gamma ^{-}\geq
0]+p_{x}(0,1,1)\times 1[-x_{31}^{\Delta }\beta -\gamma ^{+}\geq 0]$} \\ 
\hline
$3$ & $\{u:u_{32}^{\Delta }\leq 0\}$ & \multicolumn{1}{|c|}{$%
p_{x}(0,1,0)\times 1[-x_{32}^{\Delta }\beta -\gamma \leq
0]+p_{x}(1,1,0)\times 1[-x_{32}^{\Delta }\beta \leq 0]$} \\ \hline
$4$ & $\{u:u_{32}^{\Delta }\geq 0\}$ & \multicolumn{1}{|c|}{$%
p_{x}(0,0,1)\times 1[-x_{32}^{\Delta }\beta \geq 0]+p_{x}(1,0,1)\times
1[-x_{32}^{\Delta }\beta +\gamma \geq 0]$} \\ \hline
$5$ & $\{u:u_{21}^{\Delta }\leq 0\}$ & \multicolumn{1}{|c|}{$%
\left(p_{x}(1,0,0) + p_{x}(1,0,1)\right) \times 1[-x_{21}^{\Delta }\beta +\gamma ^{-}\leq
0]$} \\ 
\hline
$6$ & $\{u:u_{21}^{\Delta }\geq 0\}$ & \multicolumn{1}{|c|}{$%
\left(p_{x}(0,1,0) +p_{x}(0,1,1) \right) \times 1[-x_{21}^{\Delta }\beta -\gamma ^{-}\geq
0]$} \\ 
\hline
\end{tabular}%
\label{Table: CExch}%
\end{table}%

\subsection{Related Literature on Discrete Outcome Panel Models}\label{Sec: Related Lit Discrete Outcomes}

Analysis of binary response panel models has a long history going back to Rasch (1960, 1961),\nocite{rasch1960probabilistic} 
\nocite{rasch1961general} \cite{andersen1970asymptotic}, and see also \cite{chamberlain2010binary}, in which a static model is studied with the elements of $U$ restricted to be i.i.d. logistic, independent of $(C,X)$. With these distributional restrictions $\beta$ is
point-identified under a rank condition and consistently estimated by a
conditional maximum likelihood estimator that conditions on $\sum\limits_{t=1}^T Y_t$.\footnote{A precise statement of the rank condition is provided as Assumption 2 in 
\cite{Davezies/D'Haultfoeuille/Laage:22}.}

Extensions of panel logit models to dynamic models have been considered. \cite{honore2000panel} provides results for identification and estimation of
a dynamic binary panel model maintaining mutual independence of all elements of $U$ and
independence of $U$ and $(C,X)$, and in most cases restricting the elements of $U$ to be logistically distributed. \cite{Kitazawa:22}, \cite{Dano:23}, and \cite{Honore/Weidner:25} provide moment equations in dynamic panel logit models, which can be used to study identification and estimation of common parameters. \cite{Dobronyi/Gu/Kim/Russell:25} analyze the full
likelihood from the dynamic panel logit model and make a connection to the
truncated moment problem to obtain all of the model's observable
implications. That paper and \cite{Davezies/D'Haultfoeuille/Laage:22} also
provide characterizations of certain average and marginal effects.

An alternative to these logit specifications in the binary outcome panel model is a conditional \textit{stationarity} restriction introduced in \cite{manski1987semiparametric}, requiring that conditional on $(X,C)$ the variables $U_1,...,U_T$ all have the same marginal distribution. This is implied by the panel logit distributional restriction, but is weaker. It does not require independence of $U$ and $(X,C)$, and it can
allow for correlation in the components of $U$. Nonetheless it does restrict
the joint distribution of $U$ and $C$.\footnote{As pointed out in \cite{Chernozhukov/Fernandez-Val/Hahn/Newey:09} the
stationarity restriction $U_t|C,X \overset{d}{=}U_1|C,X$ for all $t$ is
equivalent to $(U_t,C)|X \overset{d}{=}(U_1,C)|X$ for all $t$.}

Conditional stationarity restrictions have been
used in several papers. \cite{Abrevaya:2000} studies a class of generalized regression models that nests binary response and censored outcome models under conditional stationarity and stronger restrictions. 
\cite{Chernozhukov/Fernandez-Val/Hahn/Newey:09} characterizes bounds for average
and quantile effects in several nonseparable panel models, including binary
response models. \cite{khan2023identification} provides set
identification results for common parameters in semiparametric dynamic
binary response panel models. Conditional stationarity restrictions have
also been used in multinomial response panel models, for example in \cite{shi2018estimating}, \cite
{khan2021inference}, \cite{pakes2022moment}, \cite{Pakes/Porter/Shepard/Calder-Wang:25}, \cite{Gao/Wang:24}, \cite{Gao/Li:20}, and \cite{mbakop2023identification}. 


\cite{aristodemou2021semiparametric} is the one paper of which we
are aware that studies the binary response specification \eqref{binary model}
without restricting the covariation of $C$ with either $X$ or $U$. In that
paper $U$ and $X$ are restricted to be independently distributed, in some
cases conditional on an initial condition, but,
importantly, not conditional on $C$. That paper provides bounds on
parameters in the models studied but does not claim sharpness. The
projection approach delivers characterizations of sharp identified sets and
applies more broadly, for example allowing arbitrary $T$, unobserved initial
conditions, and alternative restrictions on the joint distribution of $U$
and $X$.


The literature on fixed effects models of dynamic ordered response panels is
recent and not extensive. \cite{honore2021dynamic} employ functional
differencing to provide moment
conditions that can be used as a basis for estimation and inference in
models in which the elements of $U$ are i.i.d. logistically distributed and
independent of $X$ and $C$. In a model with these distributional
restrictions on $U$ but an alternative lag dependence specification $%
L_{t}=1[Y_{t-1}\geq k]$ for specified $k$, \cite{Muris/Raposo/Vandoros:25}
develop a conditional maximum likelihood estimator building on insights from 
\cite{honore2000panel}. References to the broader literature on ordered
response panel models, including random effects approaches and static
models, can be found in these papers. An unpublished chapter of 
\cite{Aristodemou:16} provides bounds on common parameters in
some fixed effects ordered response panel models when $T=2$.\footnote{The models studied in Chapter 6 of \cite{Aristodemou:16}, like those studied
in this paper, impose no restrictions on the joint distribution of $C$ and $U $, although the analysis requires an observed initial condition,
independence restrictions conditional on the initial condition, and does not
allow dependence on lagged latent variables $Y^{\ast }$ as is allowed here.
Moreover, non-sharp outer bounds are obtained. However, in contrast to \cite{honore2000panel} and \cite{Muris/Raposo/Vandoros:25}, in \cite{Aristodemou:16}, as here, no logistic or other parametric distributional
restriction on $U$ is required.}

\section{Numerical Illustrations of Identified Sets}\label{Sec: Numerical Illustrations}

Illustrations of identified sets for $(\beta ,\gamma )$
are presented for the CES model of Section \ref{Sec: CES model}. Both weak and strict exogeneity restrictions, MW and MS, are considered.

As in Section \ref{Sec: CES model}, the model features firm-specific unobservables $C$ and $S$ in the CES production function. For the sake of illustration, a data generation process is considered in which $T=3$, and, although it is unknown to the econometrician, $C=S=0$ for all firms, so output follows a
Cobb-Douglas specification
\[
Y_{t}=\beta _{0}\left( \gamma_0 \log (X_{t1})+(1-\gamma _{0})\log
(X_{t2})\right) +U_{t}\text{,}\qquad \quad t\in \{1,2,3\}.
\]
This is chosen for simplicity, but calculations are easily done for more complex cases.

Let $X$ denote the $3$x$2$ matrix with elements $X_{tj}$. To determine the support of $X$, values of its elements $X_{t1}$ and $X_{t2}$ were drawn i.i.d. with $\log X_{tj} \sim N(0,1/4)$ to produce $50$ support points, each of which was given equal probability.

The identified set is given by the inequalities  \eqref{CES MS inequalities 2}, equivalently for each $r$:
\begin{equation}\label{SF inequalities C2 v2}
E\left[\max_{s \in \left[-\infty,1\right]} \sum_{t=1}^T w_t(r,X) \left(Y_t  - \beta  \log g(X_{t1},X_{t2},\gamma,s) \right) \right] -
E \left[ \inf_{c \in \mathbb{R}} \sum_{t=1}^T w_t(r,X) c \right]	\geq 0\text{.}
\end{equation}
For $r$ such that $\operatorname{Pr}\left[
\sum_{t=1}^{T}w_{t}(r,X)=0\right] <1$, the last term
can be made arbitrarily large and \eqref{SF inequalities C2 v2} will be satisfied for any such $r$. Therefore we have the characterization
\begin{equation}
\min_{r : \lVert r \rVert =1} E\left[\max_{s \in \left[-\infty,1\right]} \sum_{t=1}^T w_t(r,X) \left(Y_t  - \beta  \log g(X_{t1},X_{t2},\gamma,s) \right) \right] \geq 0\text{,}	
\end{equation}
for which we need only consider values of $r$ such that $
\sum_{t=1}^{T}w_{t}(r,X)=0$ almost surely. If the support of $X$ is such that there exists no proper linear subspace of $\mathbb{R}^{2T+1}$ that contains $(1,X_1,...,X_T)$ almost surely, as is the case in this illustration, this is equivalent to imposing the restrictions
\begin{equation}\label{adding up condition}
\sum_{t=1}^T r_t = 0 \quad \text{and} \quad \sum_{t=1}^T r_{thk}=0\text{, for all } h,k\text{.}
\end{equation}	
This has the effect of removing $c$ from the
inequality.

On replacing $Y_{t}$ by its expectation conditional on $X$, denoted $\mu_t(X)$, there is:
\begin{equation}\label{SF inequalities C3}
\min_{r\in\mathcal{R}}  E \left[  \max_{s \in \left[-\infty,1\right] }\sum_{t=1}^T w_t(r,X) \biggl(\mu_t(X)-  \beta   \log g(X_{t1},X_{t2},\gamma,s) \biggr) \right]	\geq 0\text{,}
\end{equation}
where $\mathcal{R}$ is the set of $r$ such that $\Vert r \rVert = 1$ and \eqref{adding up condition} holds.

Weak and strict exogeneity restrictions are distinguished by additionally imposing $r_{thk}=0$ for all $t<h$ under weak exogeneity, as described in Section \ref{Sec: Moment Restrictions}.

The maximisation with respect to $s$ is done using the modified golden
section method provided by the \texttt{optimise} function of \textsf{R}.\footnote{\cite{R2025}: \url{https://www.R-project.org/}.} Maximisation is done with respect to two alternative monotone
transformations of $s\in (-\infty ,1]$ to the unit interval:
\[
\lambda _{1}(s)=\exp (s-1)\text{,}\quad \quad \quad \lambda _{2}(s)=\frac{1}{\pi }
\arctan (-\log (1-s))+\frac{1}{2}\text{.}
\]
Any internal maxima that are found are compared with the values obtained at $s=-\infty $ and $s=1$ and the largest value is chosen. The expectation in \eqref{SF inequalities C3} is obtained as the probability-weighted sum over the support of $X$.

Under the strict exogeneity Restriction MS stated in Section \ref{Sec: CES model}, with $T=3$ there are $21$ elements in $r$ subject to $\lVert r \rVert =1$ and the seven restrictions \eqref{adding up condition}. Under the weak exogeneity Restriction MW there are an additional six restrictions $r_{thk} = 0$ for $t < h$.\footnote{Note that imposing $r_{thk}=0$ for $t<h$ in \eqref{SF inequalities C3} corresponds to a less restrictive model than when $r_{thk}=0$ for $t<h$ is not imposed.}  Even in the case with $r$ more heavily restricted, minimization is hard to calculate precisely so an alternative calculation is done that delivers an outer region. The bounds obtained nonetheless demonstrate the informativeness of the CES model with multiple fixed effects, one entering nonlinearly.

The calculation proceeds by drawing $5000$ pseudo-random standard Gaussian
values of $r$ which are subjected to the required restrictions. A value of $(\beta ,\gamma )$ is deemed out of the identified set if the inequality \eqref{SF inequalities C3} is violated at any of the values of $r$ considered. The same
stream of $5000$ pseudo-random values of $r$ are employed as each value of $(\beta ,\gamma )$ is considered.

Figure \ref{Figure: Numerical Illustration} shows the results obtained under the weak (dark blue) and
strict (light blue) exogeneity restrictions when $(\beta _{0},\gamma
_{0})=(1,0.5)$. These outer regions are quite informative even in this simple case in which $T=3$. Having larger $T$ or taking more than $5000$ pseudo-random draws would deliver tighter bounds.

\begin{figure}[htbp]
\centering 
\includegraphics[trim = {4cm 0.5cm 0cm 0}, scale = 0.35]{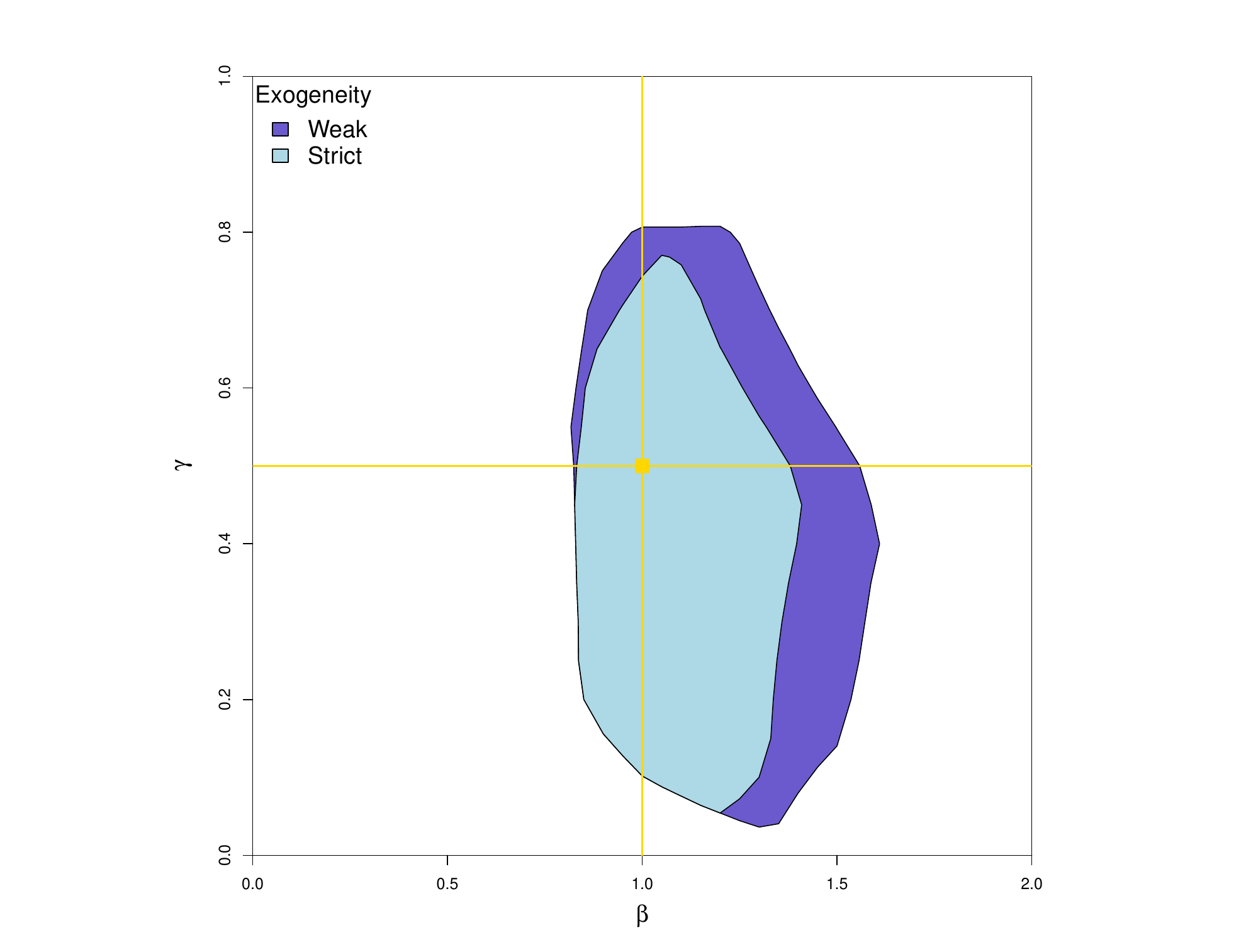}
\caption{Bounds on $\left(\beta,\gamma\right)$ in the CES Model for the data generation process described in Section \ref{Sec: Numerical Illustrations} under weak and strict exogeneity restrictions. The lines and point plotted in yellow show the values of $\beta$ and $\gamma$ in the process generating the distribution of $(Y,X)$ used in this illustration.}
\label{Figure: Numerical Illustration}
\end{figure}

\section{Discussion and concluding remarks}\label{sec: conclusion}

In the econometrics and the statistics literature the incidental parameter
problem arising with short panels has mostly been subject to analysis using particular parametric specifications of distributions of outcomes or
unobservables. Notable examples are \cite{Neyman/Scott:48}, \cite{honore2000panel}, \cite{Lancaster:00}, and \cite{Bonhomme:12}. A problem for practicing researchers is: which
distribution to choose - economic reasoning and context usually offers little guidance, and the
literature says little about the consequences of an unsuitable choice.

A notable exception is the work based on the stationarity restrictions
introduced in \cite{manski1987semiparametric}. In that work no parametric restrictions are
placed on probability distributions, but the approach is not universally applicable.

The situation in the year 2000 was summarized by Tony Lancaster as follows:\footnote{\cite{Lancaster:00}, page 404.} \textquotedblleft The absence of a method
guaranteed to work in a large class of econometric models means that any
paper on [the incidental parameter problem] must be a catalogue of
examples\textquotedblright . Little has changed in the years that followed.
The projection approach introduced in this paper fills this gap,
delivering a universally applicable solution to the incidental parameter
problem.

Taking this projection approach, incidental parameters in any number are
projected away from the space of observed and unobserved variables. The
original model specification then delivers correspondences specifying the
feasible combinations of the remaining variables. Classical
\textquotedblleft fixed effects\textquotedblright , unobserved initial
conditions and missing data are examples of variables that can be treated in
this way.

Projection delivers an incomplete model whose identifying power can be
determined by extension of available methods, for example as developed for
the analysis of Generalized Instrumental Variable models in Chesher and
Rosen (2017). Estimation and inference using the resulting characterizations
of identified sets is off-the-shelf.

With the incidental parameters projected away, robust econometric analysis
can proceed absent restrictions on their joint probability distribution with
other variables. Importantly, progress can be made using nonparametric
specifications of the distribution of the unobserved variables that vary
within observational units. For example, mean and conditional mean
restrictions can be employed as in the CES production function example of
Section \ref{Sec: continuous outcomes} and quantile independence restrictions can be used as described in
Section 4.

Four examples of application to econometric models have been set out in this
paper. More can be found in the online working paper Chesher, Rosen and
Zhang (2024), which includes applications to models admitting multiple indexes, for example, models of multiple discrete choice and simultaneous binary response.

Endogenous explanatory variables are easily accommodated following the GIV
analysis of Chesher and Rosen (2017). All endogenous variables are placed in
the list of outcomes, $Y$, and restrictions suitable for the context are
imposed on the distribution of within-observation-unit-varying $U$ and
explanatory variables, $X$. It is straightforward to impose \emph{weak}
exogeneity restrictions, for example, in dynamic panels requiring that for
all $t$, $(X_{1},\dots ,X_{t})$ and $(U_{t},\dots ,U_{T})$ satisfy some
suitable-for-context independence restriction while allowing feedback from
historic shocks and outcomes to affect the determination of future $X$
values.

Finally, the results of this paper can be useful for conducting sensitivity
analysis and specification testing. The identified sets delivered by this
paper's  models that place no restriction on the distribution of
unobservable unit-specific effects will contain the structures identified by
more restrictive models if their restrictions are satisfied by the process
under study, as captured in the distribution of observable variables the
process delivers. The analysis set out here can show how sensitive the
findings obtained using those more restrictive models are to relaxation of
their additional restrictions. It may be found that estimation employing a
point-identifying model delivers a structure outside an estimator of the
identified set obtained using a less restrictive model of the type studied
in this paper. That will suggest the more restrictive model is misspecified.
Formal development of such specification tests is a potentially fruitful 
topic for future research.

\bibliographystyle{econsoc}
\bibliography{panelgiv}

\begin{appendices}
\section{Proofs}
\noindent\textbf{Proof of Proposition \ref{theorem: profile V}}.
\textit{First to be shown is that the definitional expression of $\mathcal{I}(\mathcal{M},F_{YX})$ in \eqref{ID set definition equation} is equivalent to the expression in \eqref{Ystar selectionability}. Let $\left(f,G_{U|X}(\cdot|\cdot)\right)$ be a member of the set defined in \eqref{ID set definition equation}. Then there exist $\tilde{U}$ and $\tilde{V}$ and $\tilde{Y}$ such that for almost every $x \in \mathcal{R}_X$, conditional on $X=x$: (i) $\tilde{U} \sim G_{U|X}(\cdot|x)$, (ii) $\tilde{Y} \in \mathcal{Y}(\tilde{U},\tilde{V},X;f)$, and (iii) $\tilde{Y}\sim F_{Y|X}(\cdot|x)$. Condition (ii) implies that $\tilde{Y} \in \mathcal{Y}(\tilde{U},X;f)$. This with (i) and (iii) implies that $\left(f,G_{U|X}(\cdot|\cdot)\right)$ is an element of the set defined in \eqref{Ystar selectionability}.}

\textit{Now for the other direction let $\left(f,G_{U|X}(\cdot|\cdot)\right)$ be an element of the set defined in \eqref{Ystar selectionability}. Then there exist $\tilde{Y}$ and $\tilde{U}$ such that for almost every $x\in\mathcal{R}_X$, conditional on $X=x$: (i) $\tilde{Y}\sim F_{Y|X}(\cdot|x)$, and (ii) $\tilde{Y} \in \mathcal{Y}(\tilde{U},X;f)$, where $\tilde{U}|X=x \sim G_{U|X}(\cdot|x)$. Condition (ii) implies that for each realization of $\tilde{U}$ there exists $v\in\mathcal{R}_V$ such that $\tilde{Y}\in\mathcal{Y}(\tilde{U},v,X;f)$. This further implies that there exists a random variable $\tilde{V}$ that assigns mass to such values of $v$, i.e. $\tilde{Y}\in\mathcal{Y}(\tilde{U},\tilde{V},X;f)$. Thus $\left(f,G_{U|X}(\cdot|\cdot)\right)$ is in the set defined by \eqref{ID set definition equation}, completing the proof that \eqref{ID set definition equation} and \eqref{Ystar selectionability} are equivalent characterizations of $\mathcal{I}(\mathcal{M},F_{YX})$. The equivalence of \eqref{Ystar selectionability} and \eqref{Ustar selectionability} follows directly from relation \eqref{dual UYset relationship} and the definition of selectionability.} \qed

\begin{lemma}\label{Lemma RCS}
	Let Restriction RCS hold. For any $\left(f,G_{U|X}(\cdot|\cdot)\right)\in \mathcal{M}$, random vector $\tilde{U}$ defined on $\left( \Omega ,\mathsf{L},\mathbb{P}\right) $ with conditional distribution $G_{U|X}(\cdot|x)$ for almost every $x$  is a measurable selection of $\mathcal{U}(Y,X;f)$ if and only if the same random vector $\tilde{U}$ is a measurable selection of $\mathsf{cl}\left(\mathcal{U}(Y,X;f)\right) $.
\end{lemma}

\noindent\textbf{Proof of Lemma \ref{Lemma RCS}}.
\textit{It is immediate that if $\tilde{U}$ is a measurable selection of $\mathcal{U}(Y,X;f)$ then it is also a measurable selection of $\mathsf{cl}\left(\mathcal{U}(Y,X;f)\right)$, so only the reverse implication needs to be shown. Thus, suppose that $\tilde{U}$ is a measurable selection of $\mathsf{cl}\left(\mathcal{U}(Y,X;f)\right)$. Under Restriction RCS the event that $\tilde{U}$ is an element of $\mathsf{cl}\left(\mathcal{U}(Y,X;f)\right)$ but not an element of $\mathcal{U}(Y,X;f)$ occurs with zero probability. Since $\tilde{U}$ is a measurable selection of $\mathsf{cl}\left(\mathcal{U}(Y,X;f)\right)$ we have that $\tilde{U}\in \mathsf{cl}\left(\mathcal{U}(Y,X;f)\right)$ almost surely. It follows that $\tilde{U}\in \mathcal{U}(Y,X;f)$ almost surely and thus $\tilde{U}$ is a measurable selection of $\mathcal{U}(Y,X;f)$.} \qed
\\

\noindent\textbf{Proof of Corollary \ref{Containment Corollary}}. \textit{This follows directly from Proposition \ref{theorem: profile V} and Lemma \ref{Lemma RCS}, together with Corollary 1 of \cite{Chesher/Rosen:17}.}\qed
\\

\noindent\textbf{Proof of Proposition \ref{Theorem: CES Model Under M}}. \textit{First it will be shown that the identified set for $f$ is $\mathcal{F}_I = \left\{f\in\mathcal{F}: \mathbf{0} \in \mathbb{E}_I\left[\mathcal{Q}(Y,X;f)\right]\right\}$. Under Restriction M, $\mathcal{Q}(Y,X;\theta)$ has an integrable selection and is closed, so it is an integrable random closed set, and there exists random vector $Q$ such that $Q \in \mathcal{Q}(Y,X;f)$ almost surely, i.e. a measurable selection of $\mathcal{Q}(Y,X;f)$, with mean $\mathbf{0}_J$ if and only if $\mathbf{0}_J \in \mathbb{E}_I\left[\mathcal{Q}(Y,X;f)\right]$. By definition of $ \mathcal{Q}(Y,X;f)$, the set of measurable selections of $\mathcal{Q}(Y,X;f)$ is the set of random vectors $(Z_1^{\prime}\tilde{U}_1,...,Z_T^{\prime}\tilde{U}_T)$ such that $\tilde{U}$ is a measurable selection of $ \mathcal{U}(Y,X;f)$. Thus there exists $Q \in \mathcal{Q}(Y,X;f)$ almost surely with $E[Q] = \mathbf{0}_J$ if and only if there exists $\tilde{U} \in \mathcal{U}(Y,X;f)$ almost surely with $E[Z_t^{\prime}\tilde{U}_t] = \mathbf{0}_{J_t}$ for all $t$. Thus $\mathcal{F}_I$ is the identified set for $f$.} 

\textit{The moment closure of $\mathcal{F}_I$ is $\overline{\mathcal{F}_I} = \left\{f\in\mathcal{F}: \mathbf{0} \in \mathbb{E}\left[\mathcal{Q}(Y,X;f)\right]\right\}$, where $\mathbb{E}\left[\mathcal{Q}(Y,X;f)\right] = \mathsf{cl}\left(\mathbb{E}_I\left[\mathcal{Q}(Y,X;f)\right]\right)$. Under Restriction PM, $\left( \Omega,\mathsf{L},\mathbb{P}\right)$ is nonatomic and by Theorem 2.1.26 of \cite{Molchanov:17} we have that $\mathbb{E}\left[\mathcal{Q}(Y,X;\theta)\right]$ is convex. Because $\mathbb{E}\left[\mathcal{Q}(Y,X;f)\right]$ is convex, $\mathbf{0}_J \in \mathbb{E}\left[\mathcal{Q}(Y,X;f)\right]$ if and only if $0 \leq h(\mathbb{E}\left[\mathcal{Q}(Y,X;\theta)\right],r)$ for all $r \in \mathbb{R}^J$. Then \eqref{moment closure set} follows because the support function is positive homogeneous and  $h(\mathbb{E}\left[\mathcal{Q}(Y,X;\theta)\right],r) = E\left[h(\mathcal{Q}(Y,X;\theta),r) \right]$ by Theorem 2.1.35 of \cite{Molchanov:17}. Finally, if $\mathbb{E}_I\left[\mathcal{Q}(Y,X;f)\right]$ is closed then it is equal to $\mathbb{E}\left[\mathcal{Q}(Y,X;f)\right]$ and consequently $\mathcal{F}_I = \overline{\mathcal{F}_I}$.} \qed

\begin{proposition}\label{Prop CES sharpness}
	Let the restrictions of Proposition \ref{Theorem: CES Model Under M} hold with the CES specification in \eqref{CES model for y}. Suppose in addition that
	\begin{equation} \label{E1 bound}
	\forall t\in [T],\quad E\left[\sup_{s \in [-\infty,1]}\lVert Z_t\left(Y_t-\beta \log g(X_{t1},X_{t2},\gamma,s) \right) \rVert \right]	< \infty\text{.}
	\end{equation}
Then $\mathbb{E}_I\left[\mathcal{Q}(Y,X;\theta) \right]$ is closed, and the moment closure of the identified set, $\Theta^{\ast}$, is sharp.
\end{proposition}

\noindent\textbf{Proof of Proposition \ref{Prop CES sharpness}}.
\textit{By definition}
\begin{equation*}
	\mathbb{E}_I\left[ \mathcal{Q}(Y,X;\theta) \right] = \left\{q \in \mathbb{R}^J: q= E\left[Q\right] \text{ for some $Q \in \mathbf{L}^1(\mathcal{Q}(Y,X;\theta))$ } \right\}\text{.}
\end{equation*}
Let $A_t(\theta,s) \equiv Z_t\left(Y_t-\beta \log g(X_{t1},X_{t2},\gamma,s) \right)$, $A(\theta,s)\equiv(A_1(\theta,s),...,A_T(\theta,s)) $, and $J$-element vector $Z \equiv (Z_1,...,Z_T)$. Then $Q \in \mathbf{L}^1(\mathcal{Q}(Y,X;\theta))$ implies that for some measurable $(S,C)$ with support in $[-\infty,1] \times \mathbb{R}$, $Q=A(\theta,S)-CZ.$ Since $E\left[A(\theta,S) \right]$ is bounded by \eqref{E1 bound}, $E\left[Q\right] = E\left[A(\theta,S)\right] - E\left[CZ\right]$, and $\mathbb{E}_I\left[ \mathcal{Q}(Y,X;\theta) \right] = \mathcal{E}_1 + \mathcal{E}_2$ where
\begin{align*}
	\mathcal{E}_1 &= \left\{E\left[A(\theta,S) \right]:S\in [-\infty,1] \text{ and } S \text{ measurable} \right\} = \mathbb{E}_I\left[\left\{A(\theta,s):s\in [-\infty, 1] \right\} \right] \text{,} \\
	\mathcal{E}_2 &= \left\{E\left[-C Z \right]:E\left[\lvert C \rvert \lVert Z \rVert \right] < \infty   \text{ and } C \text{ measurable} \right\}\text{.}
\end{align*}
The set $\left\{A(\theta,s):s\in [-\infty, 1] \right\}$ is an integrably bounded random compact set by \eqref{E1 bound}, so $\mathcal{E}_1$ is compact by Theorem 2.1.38 of \cite{Molchanov:17}. Both $\mathcal{E}_1$ and $\mathcal{E}_2$  are convex because the underlying probability space is nonatomic by Theorem 2.1.26 of \cite{Molchanov:17}. The set $ \mathcal{E}_2$ is a linear subspace of $\mathbb{R}^J$ and is therefore closed. The sum of a compact convex set and a closed convex set is closed, see e.g. Corollary 9.1.2 of \cite{Rockafellar:70}, so $\mathbb{E}_I\left[ \mathcal{Q}(Y,X;\theta) \right]$ is closed and equal to $\mathbb{E}\left[ \mathcal{Q}(Y,X;\theta) \right]$, completing the proof. \qed

\end{appendices}

\end{document}